\documentclass[11pt]{article}
\usepackage{longtable}
\usepackage{caption}
\usepackage{multirow,multicol}
\usepackage{epsf, graphicx}
\usepackage{latexsym,amsfonts,amsbsy,amssymb}
\usepackage{amsmath,amsthm}
\usepackage{cite}
\usepackage{cleveref}
\usepackage{indentfirst}
\usepackage{bm}
\usepackage [latin1]{inputenc}
\usepackage{enumitem}
\setenumerate[1]{itemsep=0pt,partopsep=0pt,parsep=\parskip,topsep=5pt}
\setitemize[1]{itemsep=0pt,partopsep=0pt,parsep=\parskip,topsep=5pt}
\setdescription{itemsep=0pt,partopsep=0pt,parsep=\parskip,topsep=5pt}
\allowdisplaybreaks[4]
\makeatletter

\@addtoreset{equation}{section} \makeatother
\newtheorem{Theorem}{Theorem}[section]
\newtheorem{Lemma}{Lemma}[section]
\newtheorem{Corollary}{Corollary}[section]
\newtheorem{Remark}{Remark}[section]

\newtheorem{Proposition}{Proposition}[section]
\newtheorem{Example}{Example}[section]
\makeatletter \@addtoreset{equation}{section} \makeatother
\begin{document}
	
	\title{Two classes of LCD BCH codes over finite fields \let\thefootnote\relax\footnotetext{E-Mail addresses: yuqingfu@mails.ccnu.edu.cn (Y. Fu), hwliu@ccnu.edu.cn (H. Liu).}}
	\author{ Yuqing Fu,~ Hongwei Liu}
	\date{\small School of Mathematics and Statistics, Central China Normal University, Wuhan, 430079, China}
	\maketitle
	{\noindent\small{\bf Abstract:} BCH codes form an important subclass of cyclic codes, and are widely used in compact discs, digital audio tapes and other data storage systems to improve data reliability. As far as we know, there are few results on $q$-ary BCH codes of length $n=\frac{q^{m}+1}{q+1}$. This is because it is harder to deal with BCH codes of such length. In this paper, we study $q$-ary BCH codes with lengths $n=\frac{q^{m}+1}{q+1}$ and $n=q^m+1$. These two classes of BCH codes are always LCD codes. For $n=\frac{q^{m}+1}{q+1}$, the dimensions of narrow-sense BCH codes of length $n$ with designed distance $\delta=\ell q^{\frac{m-1}{2}}+1$ are determined, where $q>2$ and $2\leq \ell \leq q-1$. Moreover, the largest coset leader is given for $m=3$ and the first two largest coset leaders are given for $q=2$. The parameters of BCH codes related to the first few largest coset leaders are investigated. Some binary BCH codes of length $n=\frac{2^m+1}{3}$ have optimal parameters. For ternary narrow-sense BCH codes of length $n=3^m+1$, a lower bound on the minimum distance of their dual codes is developed, which is good in some cases.}
	
	\vspace{1ex}
	{\noindent\small{\bf Keywords:}
		BCH codes; LCD codes; reversible codes; cyclic codes; dual codes.}
	
	2020 \emph{Mathematics Subject Classification}:  94B05, 94B60
	\section{Introduction}
	Let $\mathbb{F}_{q}$ be the finite field with $q$ elements, where $q$ is a prime power. An $[n,k,d]$ linear code 
	$\mathcal{C}$ over $\mathbb{F}_{q}$ is a $k$-dimensional subspace of $\mathbb{F}_{q}^{n}$ with minimum (Hamming) distance $d$. The (Euclidean) dual code of $\mathcal{C}$, denoted by $\mathcal{C}^{\perp}$, is defined as 
	$$\mathcal{C}^{\perp}=\big\{{\bf x}\in \mathbb{F}_{q}^{n}~\big|~{\bf x}\cdot {\bf c}=0~{\rm for~all}~{\bf c}\in \mathcal{C}\big\},$$
	where ${\bf x}\cdot {\bf c}$ denotes the standard inner product of the two vectors ${\bf x}$ and ${\bf c}$. A linear code $\mathcal{C}$ is said to be cyclic if $(c_{0},c_{1},\cdots,c_{n-1})\in \mathcal{C}$ implies $(c_{n-1},c_{0},\cdots,c_{n-2})\in \mathcal{C}$. As usual, $\mathbb{F}_{q}[x]$ denotes the polynomial ring over $\mathbb{F}_{q}$ and $\mathcal{R}_{n}:=\mathbb{F}_{q}[x]/\langle x^{n}-1\rangle$ denotes the quotient ring of $\mathbb{F}_{q}[x]$ with respect to the ideal $\langle x^{n}-1\rangle$. By identifying each vector $(a_{0},a_{1},\cdots,a_{n-1})\in \mathbb{F}_{q}^{n}$ with a polynomial $$a_{0}+a_{1}x+\cdots+a_{n-1}x^{n-1}\in \mathcal{R}_{n},$$ a linear code $\mathcal{C}$ of length $n$ over $\mathbb{F}_{q}$ corresponds to an $\mathbb{F}_{q}$-subspace of the ring $\mathcal{R}_{n}$. Moreover, $\mathcal{C}$ is cyclic if and only if the corresponding subspace is an ideal of $\mathcal{R}_{n}$. Let $\mathcal{C}$ be a cyclic code of length $n$ over $\mathbb{F}_{q}$, then $\mathcal{C}=\langle g(x)\rangle$, where $g(x)$ is a monic divisor of $x^{n}-1$ and has the minimum degree among all the generators of $\mathcal{C}$. Such polynomial $g(x)$ is unique and is called the generator polynomial of $\mathcal{C}$. 
	
	Let $n$ be a positive integer coprime to $q$. Denote $r={\rm ord}_{n}(q)$, i.e., $r$ is the smallest positive integer such that $q^{r}\equiv 1~({\rm mod}~n)$. Let $\alpha$ be a primitive element of $\mathbb{F}_{q^{r}}$ and put $\beta=\alpha^{\frac{q^{r}-1}{n}}$. Then $\beta$ is a primitive $n$-th root of unity. The set $T=\{0\leq i\leq n-1~|~g(\beta^{i})=0\}$ is referred to as the defining set of $\mathcal{C}$. For integers $b$ and $\delta$ with $2\leq \delta \leq n$, let $\mathcal{C}_{(q,n,\delta,b)}$ denote the cyclic code of length $n$ over $\mathbb{F}_{q}$ with defining set 
	$$T=C_{b}\cup C_{b+1}\cup\cdots\cup C_{b+\delta-2},$$
	where $C_{i}$ is the $q$-cyclotomic coset modulo $n$ containing $i$. The code $\mathcal{C}_{(q,n,\delta,b)}$ is called a BCH code with designed distance $\delta$. When $b=1$, $\mathcal{C}_{(q,n,\delta,1)}$ is called a narrow-sense BCH code. The code $\mathcal{C}_{(q,n,\delta+1,0)}$ is exactly the even-like subcode of $\mathcal{C}_{(q,n,\delta,1)}$. The BCH codes with length $n=q^m-1$, $\frac{q^m-1}{q-1}$ and $q^m+1$ are called primitive BCH codes, projective BCH codes and  antiprimitive BCH codes, respectively.
	
	BCH codes form a special class of cyclic codes and are widely used in compact discs, digital audio tapes and other data storage systems to improve data reliability. Among all types of BCH codes, primitive BCH codes are the most extensively studied. The well-known Reed-Solomon codes are primitive BCH codes with length $n=q-1$. The projective BCH codes are also interesting. Many Hamming codes are projective narrow-sense BCH codes. In general, it is difficult to determine the parameters such as the dimension, the minimum distance and the weight distribution of BCH codes, let alone the parameters of the dual codes of BCH codes. Very recently, \cite{gdl}, \cite{wwlw} and \cite{zlx} gave sufficient and necessary conditions in terms of the designed distance $\delta$ for the dual codes of BCH codes with lengths $n=\frac{q^m-1}{q-1}$, $\frac{q^m-1}{q+1}$ and $\frac{q^{2m}-1}{q+1}$ to be BCH codes. \cite{gdl} and \cite{wwlw} also developed some lower bounds on the minimum distances of $\mathcal{C}_{(q,(q^m-1)/(q-1),\delta,1)}^{\perp}$ and $\mathcal{C}_{(q,(q^m-1)/(q+1),\delta,1)}^{\perp}$.
	
	A linear code $\mathcal{C}$ over $\mathbb{F}_{q}$ is called an LCD (linear complementary dual) code provided $\mathcal{C}\cap \mathcal{C}^{\bot}=\{{\bf 0}\}$. LCD cyclic codes over finite fields were referred to as reversible codes in the literature and were first studied by Massey for the data storage applications \cite{m1}. There has been a lot of research on LCD codes due to their important application in cryptography to resist side-channel attacks and fault non-invasive attacks (see \cite{bccgh, cg} for more details). The constructions and analysis of LCD BCH codes with lengths $n=q^m-1$ and $n=\frac{q^m-1}{q-1}$ were presented in \cite{hywsm,hyws,ldl1,lldl}. It was shown in \cite{ldl1} that a cyclic code of length $n$ over $\mathbb{F}_{q}$ is an LCD code if $-1$ is a power of $q$ modulo $n$. In particular, cyclic codes over $\mathbb{F}_{q}$ with length $n=\frac{q^m+1}{N}$ are always LCD codes, where $1\leq N<q^m+1$ is a divisor of $q^m+1$.
	
   Determining the parameters of BCH codes over $\mathbb{F}_{q}$ with length $n=\frac{q^m+1}{N}$ has become an interesting area of research. For $n=q^m+1$, the authors in \cite{ldl1, ldl2, llgs, llfl, ylly, zc, zswh} gave necessary and sufficient conditions for an integer $1\leq a\leq q^m$ being a coset leader and considered the first few largest coset leaders. Based on these results, the parameters of antiprimitive BCH codes with designed distance $\delta$ in some ranges are investigated. Zhu et al. \cite{zswh} gave some lower bounds on the minimum distance of antiprimitive BCH codes. To the best of our knowledge, there are only two papers \cite{ylly, z} on BCH codes with length $n=\frac{q^m+1}{q+1}$. In \cite{ylly} and \cite{z}, the authors found out all coset leaders in the range $[1,q^{\frac{m-1}{2}}]$ and $[q^{\frac{m-1}{2}}+1, \frac{2q^{\frac{m+1}{2}}+2q-1}{q+1}]$, respectively, and then settled the dimensions of $\mathcal{C}_{(q,n,\delta,1)}$ for $2\leq \delta\leq q^{\frac{m-1}{2}}+1$ and $q^{\frac{m-1}{2}}+2 \leq \delta\leq \frac{2q^{\frac{m+1}{2}}+2q-1}{q+1}+1$, respectively. Zhang \cite{z} also determined the first two largest coset leaders for $q=3$ and explored the parameters of related BCH codes. In \cite{zc, zlks, z}, the dimensions of BCH codes with length $n=\frac{q^m+1}{N}$ were  studied, where $N$ is a divisor of $q+1$ with $1<N<q+1$.
	
   In this paper, we work on BCH codes over $\mathbb{F}_{q}$ with lengths $n=\frac{q^m+1}{q+1}$ and $n=q^m+1$. We will investigate the dimension and minimum distance of $\mathcal{C}_{(q,(q^m+1)/(q+1),\delta,1)}$ for certain $\delta$ and give a lower bound on the minimum distance of $\mathcal{C}_{(3,3^m+1,\delta,1)}^{\perp}$. The rest of this paper is organized as follows. In Section 2, we present some notions and results that will be used later. In Section 3, let $n=\frac{q^m+1}{q+1}$, we first find out all coset leaders in the range $[q^{\frac{m-1}{2}}+1,(q-1)q^{\frac{m-1}{2}}]$ and then provide formulas on the dimensions of $\mathcal{C}_{(q,n,\delta,1)}$ for $\delta=\ell q^{\frac{m-1}{2}}+1$, where $q>2$ and $2\leq \ell \leq q-1$. As a byproduct, we find out the largest coset leader modulo $n=\frac{q^3+1}{q+1}$. The parameters of related BCH codes are then investigated. In Section 4, for $q=2$ and $n=\frac{2^m+1}{3}$, we determine the first two largest coset leaders and explore the parameters of binary BCH codes with small dimensions. Some optimal codes are found here. In Section 5, we establish a lower bound on the minimum distance of $\mathcal{C}_{(3,3^m+1,\delta,1)}^{\perp}$. Some examples show that this lower bound is good. Section 6 concludes this paper. We remark that all examples in this paper are verified by C or Magma programs.

	\section{Preliminaries}
	Throughout this paper, $\mathbb{F}_{q}$ denotes the finite field with $q$ elements, where $q$ is a prime power and $n$ is a positive integer coprime to $q$. Let $\mathbb{Z}_{n}=\{0,1,2,\cdots,n-1\}$ denote the residue ring of the integer ring $\mathbb{Z}$ modulo $n$, and let $\mathbb{Z}_{n}^{*}$ be the group of units in $\mathbb{Z}_{n}$. For $a\in \mathbb{Z}_{n}^{*}$, ${\rm ord}_{n}(a)$ denotes the order of $a$ in $\mathbb{Z}_{n}^{*}$. As usual, $|X|$ stands for the cardinality of a finite set $X$. Given two integers $a_{1}$ and $a_{2}$, if $a_{1}$ divides $a_{2}$, then we write $a_{1}|a_{2}$; ${\rm gcd}(a_{1},a_{2})$ denotes the greatest common divisor of $a_{1}$ and $a_{2}$. Given a rational number $x$, $\lceil x\rceil$ denotes the smallest integer greater than or equal to $x$ and $\lfloor x\rfloor$ denotes the largest integer less than or equal to $x$.
	
	For any $a\in \mathbb{Z}_{n}$, the $q$-cyclotomic coset ($q$-coset in short) of $a$ modulo $n$ is defined by
	$$C_{a}=\{a,aq,aq^{2},\cdots,aq^{l_{a}-1}\}~{\rm mod}~n~\subseteq \mathbb{Z}_{n},$$
	where $l_{a}$ is the least positive integer such that $aq^{l_{a}}\equiv a~({\rm mod}~n)$, and is the size of $C_{a}$. It is known that $l_{a}|{\rm ord}_{n}(q)$. For $1\leq t\leq {\rm ord}_{n}(q)$, $[aq^{t}]_{n}$ denotes the integer in $\mathbb{Z}_{n}$ which is congruent to $aq^{t}$ modulo $n$. The smallest element in $C_{a}$ is called the coset leader of $C_{a}$. Let $\Gamma_{(n,q)}$ be the set of all $q$-coset leaders modulo $n$. In this paper, we denote by $\delta_{i}$ the $i$-th largest coset leader in $\Gamma_{(n,q)}$. All the distinct $q$-cosets modulo $n$ partition $\mathbb{Z}_{n}$. The defining set $T$ of a cyclic code $\mathcal{C}$ of length $n$ over $\mathbb{F}_{q}$ is the union of some $q$-cosets modulo $n$, and the dimension of $\mathcal{C}$ is equal to $n-|T|$.
		
	It is easily seen that the dimension of the narrow-sense BCH code $\mathcal{C}_{(q,n,\delta,1)}$ is given by
	$${\rm dim}(\mathcal{C}_{(q,n,\delta,1)})=n-|C_{1}\cup C_{2}\cup\cdots \cup C_{\delta-1}|.$$
	Therefore, to determine the dimension of $\mathcal{C}_{(q,n,\delta,1)}$, we need to find out all coset leaders in the range $[1,\delta-1]$ and compute the cardinalities of the $q$-cosets containing coset leaders. 
	
    The following result will be used to determine the dimension of BCH codes.
		
		\begin{Lemma}\label{l3.3}
			Let $u,v$ be non-negative integers and let $l$ be a positive integer. Then
			\begin{align*}
				{\rm gcd}(l^{u}+1,l^{v}-1)=\begin{cases}
					1, & {\rm if}~\frac{v}{{\rm gcd}(u,v)}~{\rm is~odd~and}~l~{\rm is~even},\\
					2, & {\rm if}~\frac{v}{{\rm gcd}(u,v)}~{\rm is~odd~and}~l~{\rm is~odd},\\
					l^{{\rm gcd}(u,v)}+1, & {\rm if}~\frac{v}{{\rm gcd}(u,v)}~{\rm is~even}.
				\end{cases}
			\end{align*}
		\end{Lemma}
	
	The $q$-coset leaders modulo $n=\frac{q^m+1}{N}$ and the $q$-coset leaders modulo $n=q^m+1$ have the following relationship.
	\begin{Lemma}{\rm \cite{zc}}
		Suppose $q$ is a prime power and let $N$ be a divisor of $q+1$. For every integer $s$ with $1\leq s< \frac{q^m+1}{N}$, $s$ is the coset leader of the $q$-coset $C_{s}$ modulo $\frac{q^m+1}{N}$ if and only if $Ns$ is the coset leader of the $q$-coset $C_{Ns}$ modulo $q^m+1$. Moreover, $|C_{s}|=|C_{Ns}|$.
	\end{Lemma}
	
	Given a positive integer $s$ with $1\leq s\leq q^{m}-1$, we can write $$s=s_{m-1}q^{m-1}+s_{m-2}q^{m-2}+\cdots+s_{1}q+s_{0},~~0\leq s_{i}\leq q-1~{\rm for}~0\leq i\leq m-1,$$ 
	which is called the $q$-adic expansion of $s$.
	This defines a sequence $(s_{m-1},s_{m-2},\cdots,s_{0})$. In studying the first few largest coset leaders modulo $n=\frac{q^{m}+1}{N}$, it is convenient to allow each $s_{i}$ to be any integer and identify the integer $s$ with its sequence forms. For example, take $q=2$, $m=5$ and $s=11$. It is easy to see that $$s=2^{3}+2+1=2^{4}+(-1)\cdot 2^2+(-1)\cdot 1=2^{2}+2\cdot2+3\cdot 1.$$
	Then we write
	$$s=(0,1,0,1,1),~(1,0,-1,0,-1),~{\rm or}~(0,0,1,2,3).$$
	For convenience, we write $q^{m}+1=(\overset{m-1}{\overbrace{q-1,q-1,\cdots,q-1}},q+1)$. 
	
	Given two positive integers $s$ and $s'$ with $1\leq s,s'\leq q^{m}-1$. Suppose $s=\sum_{i=0}^{m-1}s_{i}q^{i}$ and $s'=\sum_{i=0}^{m-1}s'_{i}q^{i}$, where $s_{i},s_{i}'\in \mathbb{Z}$. Write $s$ and $s'$ as $$s=(s_{m-1}, s_{m-2},\cdots,s_{0})~{\rm and}~s'=(s'_{m-1}, s'_{m-2},\cdots,s'_{0}).$$ When writing $$(s_{m-1}, s_{m-2},\cdots,s_{0})>({\rm resp.} =,\geq)~(s'_{m-1}, s'_{m-2},\cdots,s'_{0}),$$ $$(s_{m-1}, s_{m-2},\cdots,s_{0})>({\rm resp.} =,\geq)~s',~{\rm or}~s>({\rm resp.} =,\geq)~(s'_{m-1}, s'_{m-2},\cdots,s'_{0}),$$ we always mean $s>({\rm resp.} =,\geq)~s'$. It is easily seen that $$s\pm s'=(s_{m-1}\pm s'_{m-1}, s_{m-2}\pm s'_{m-2},\cdots,s_{0}\pm s'_{0}).$$ 
	
	There is a well-known bound on the minimum distance of cyclic codes.
		\begin{Lemma}{\rm \bf(BCH bound)}\label{l2.3}
		Let $\mathcal{C}$ be a cyclic code of length $n$ over $\mathbb{F}_{q}$ with defining set $T$. Assume $T$ contains $\delta-1$ consecutive elements for some integer $\delta$. Then $\mathcal{C}$ has minimum distance $d\geq \delta$. 
	\end{Lemma}
	
	  The lemmas below will be employed later.
	
	\begin{Lemma}\label{l2.4}{\rm \cite{bb}}
		The code $\mathcal{C}_{(q,n,\delta,1)}$ has minimum distance $d=\delta$ if $\delta|n$. 
	\end{Lemma}
	
	\begin{Lemma}\label{l2.5}
		Suppose $n=\frac{q^m+1}{N}$, where $1\leq N<q^m+1$ is a divisor of $q^m+1$. The code $\mathcal{C}_{(q,n,\delta+1,0)}$ has minimum distance $d\geq 2\delta$. In particular, $d=2\delta$ if $\delta\big| n$.
	\end{Lemma}
	\begin{proof}
		Since $q^m\equiv -1~({\rm mod}~n)$, the defining set of $\mathcal{C}_{(q,n,\delta+1,0)}$ contains $\{-\delta+1,-\delta+2,\cdots,-1,0,1,\cdots,\delta-2,\delta-1\}$ as a subset. It follows from the BCH bound that $d\geq 2\delta$. If $\delta|n$, by Lemma \ref{l2.4} there exists a codeword $c(x)\in \mathcal{C}_{(q,n,\delta,1)}$ with weight $\delta$. Note that $(x-1)c(x)\in \mathcal{C}_{(q,n,\delta+1,0)}$ and that the weight of $(x-1)c(x)$ is at most $2\delta$, thus $d\leq 2\delta$, which implies $d= 2\delta$.
	\end{proof}
	
	\section{Dimensions of $q$-ary BCH codes with length $\frac{q^m+1}{q+1}$}
	Throughout this section, $n=\frac{q^{m}+1}{q+1}$, where $m$ can only be an odd integer. The goal of this section is to determine the dimension of the narrow-sense BCH code $\mathcal{C}_{(q,n,\delta,1)}$ for $\delta=\ell q^{\frac{m-1}{2}}+1$, where $q>2$ and $2\leq \ell\leq q-1$. 
	
	The earlier reference gave the following two results.
	\begin{Lemma}\label{l3.1}{\rm \cite{ylly}}
		${\rm ord}_{n}(q)=2m$ if $(q,m)\neq (2,3)$.
	\end{Lemma}
	
	\begin{Lemma}\label{l3.2}{\rm \cite{ylly}}
		For $1\leq a\leq q^{\frac{m-1}{2}}$ with $a\not\equiv0~({\rm mod}~q)$, $a$ is a coset leader with $|C_{a}|=2m$ except that $a=\frac{q^{\frac{m+1}{2}}-(-1)^{\frac{m+1}{2}}}{q+1}$.
	\end{Lemma}
	In this section, we consider integer $a$ in the larger range $1\leq a\leq (q-1)q^{\frac{m-1}{2}}$. We first compute the cardinalities of the cosets containing coset leaders in the range $[1,(q-1)q^{\frac{m-1}{2}}]$. 
	\begin{Proposition}\label{p3.1}
		 Suppose that $(q,m)\neq (2,3)$ and that $a$ is a coset leader in the range $1\leq a\leq q^{\frac{m+1}{2}}$ {\rm (}$1\leq a\leq q^{2}-q$ if $m=3${\rm )}, then $|C_{a}|=2m$ except in the following cases:
		
		\noindent {\rm (1)} $q\equiv 2~({\rm mod}~3)$, $m=3$ and $a=\frac{q^{2}-q+1}{3}$. Moreover, $|C_{a}|=2$.
		
		\noindent {\rm (2)} $q=2$, $m=9$ and $a=19$. Moreover, $|C_{a}|=6$.
		
		\noindent {\rm (3)} $q=4$, $m=5$ and $a=41$. Moreover, $|C_{a}|=2$. 
	\end{Proposition}
	
	\begin{proof}
		We see from Lemma \ref{l3.1} that ${\rm ord}_{n}(q)=2m$. Since $|C_{a}|$ is a divisor of $2m$, $|C_{a}|=2m$, $m$, $\frac{2m}{l}$, or $\frac{m}{l}$, where $3\leq l\leq m$ and $l|m$.
		
		{\bf Case 1.} Suppose $|C_{a}|=m$. Then $-a\equiv aq^{m}\equiv a~({\rm mod}~n)$ implying $n|2a$. As $n=\frac{q^{m}+1}{q+1}=(q-1)q^{m-2}+(q-1)q^{m-4}+\cdots+(q-1)q+1$, $n$ is odd; thus $n|a$, which is impossible.
		
		{\bf Case 2.} Suppose $|C_{a}|=\frac{2m}{l}$. Then $n\big|a(q^{\frac{2m}{l}}-1)$, or equivalently, $(q^{m}+1)\big|a(q+1)(q^{\frac{2m}{l}}-1)$. By 
		Lemma \ref{l3.3}, ${\rm gcd}(q^{m}+1,q^{\frac{2m}{l}}-1)=q^{\frac{m}{l}}+1$, and so $(q^{m}+1)\big|a(q+1)(q^{\frac{m}{l}}+1)$. Note that
		\begin{equation}\label{e3.1}
			a(q+1)(q^{\frac{m}{l}}+1)\leq q^{\frac{m+1}{2}}(q+1)(q^{\frac{m}{l}}+1)=q^{\frac{(m+3)l+2m}{2l}}+q^{\frac{(m+1)l+2m}{2l}}+q^{\frac{m+3}{2}}+q^{\frac{m+1}{2}}.
		\end{equation}
		Since $l|m$ and $m$ is odd, three subcases arise: $l=3$, $5\leq l\leq \frac{m}{3}$, and $l=m$. 
		
		{\bf Subcase 2.1.} Assume $l=3$. It follows from (\ref{e3.1}) that $a(q+1)(q^{\frac{m}{3}}+1)\leq q^{\frac{5m+9}{6}}+q^{\frac{5m+3}{6}}+q^{\frac{m+3}{2}}+q^{\frac{m+1}{2}}$. Thus $(q^{m}+1)\big|a(q+1)(q^{\frac{m}{3}}+1)$ holds only if $m\leq \frac{5m+9}{6}$, i.e., $m\leq 9$. As $l|m$ and $m$ is odd, $m=3$ or $9$.
		
	    If $m=3$, then $n=\frac{q^{3}+1}{q+1}$ and $n|a(q+1)$. As $n=(q+1)(q-2)+3$, ${\rm gcd}(n,q+1)=1$ or $3$. If ${\rm gcd}(n,q+1)=1$, then $n|a$, which is impossible. If ${\rm gcd}(n,q+1)=3$, or equivalently, $q\equiv 2~({\rm mod}~3)$, then $\frac{n}{3}|a$; thus $a=\frac{n}{3}$ and $|C_{a}|=\frac{2\times 3}{3}=2$.
		
		If $m=9$, then $(q^{9}+1)\big|a(q+1)(q^{3}+1)$. We see from (\ref{e3.1}) that $a(q+1)(q^{3}+1)\leq q^{9}+q^{8}+q^{6}+q^{5}<2(q^{9}+1)$, and hence $a(q+1)(q^{3}+1)=q^{9}+1$. So $(q+1)\big|\frac{q^{9}+1}{q^{3}+1}$. Note that $\frac{q^{9}+1}{q^{3}+1}=(q+1)(q^{5}-q^{4}+q^{3}-2q^{2}+2q-2)+3$.
		Thus $(q+1)|3$ implying $q=2$. Meanwhile, $a=\frac{2^{9}+1}{3(2^{3}+1)}=19$ and $|C_{a}|=\frac{2\times 9}{3}=6$.
		
		{\bf Subcase 2.2.} Assume $5\leq l\leq \frac{m}{3}$. By the inequality (3.1), we have
		\begin{align*}
			a(q+1)(q^{\frac{m}{l}}+1)&<q^{\frac{(m+5)l+2m}{2l}}=q^{\frac{(l+2)m+5l}{2l}}<q^{\frac{(l+2)m+9l}{2l}}\leq q^{\frac{(l+2)m+3m}{2l}}\\
			&=q^{\frac{(l+5)m}{2l}}\leq q^{\frac{2lm}{2l}}=q^{m}<q^{m}+1,
		\end{align*}
		which is a contradiction to the fact that $(q^{m}+1)\big|a(q+1)(q^{\frac{m}{l}}+1)$.
		
		{\bf Subcase 2.3.} Assume $l=m$. Then $(q^{m}+1)\big|a(q+1)^{2}$. It follows from (3.1) that $a(q+1)^{2}\leq q^{\frac{m+5}{2}}+2q^{\frac{m+3}{2}}+q^{\frac{m+1}{2}}\leq q^{\frac{m+7}{2}}+q^{\frac{m+1}{2}}$,
		and so $(q^{m}+1)\big|a(q+1)^{2}$ holds only if $m\leq \frac{m+7}{2}$, i.e., $m\leq 7$; thus $m=3$, $5$ or $7$ as $m\geq 3$ is odd. The case $m=3$ has been discussed in Subcase 2.1. 
		
	    If $m=5$, then $n=\frac{q^{5}+1}{q+1}$ and $n\big|a(q+1)$. As $n=(q+1)(q^{3}-2q^{2}+3q-4)+5$, ${\rm gcd}(n,q+1)=1$ or $5$. If ${\rm gcd}(n,q+1)=1$, then $n|a$, which is impossible. Suppose ${\rm gcd}(n,q+1)=5$, then $\frac{n}{5}\big|a$. Hence $a=\frac{\ell n}{5}$, where $1\leq \ell \leq 4$. The assumption $a\leq q^{\frac{m+1}{2}}=q^{3}$ gives $(\ell q^{2}-5q-5)q^{3}+\ell\leq 0$, which holds if and only if $\ell q^{2}-5q-5<0$, yielding $q=4$ and $\ell=1$. Thus $a=\frac{4^{5}+1}{5\times 5}=41$ and $|C_{a}|=\frac{2\times 5}{5}=2$.
		
		If $m=7$, then $n=\frac{q^{7}+1}{q+1}$ and $(q^{7}+1)\big|a(q+1)^{2}$. We know from (3.1) that $a(q+1)^{2}\leq q^{7}+q^{4}<2(q^{7}+1)$, and hence $a(q+1)^{2}=q^{7}+1$. So $(q+1)|n$. Since $n=(q+1)(q^{5}-2q^{4}+3q^{3}-4q^{2}+5q-6)+7$, 
		$(q+1)|7$. Thus $q=6$ and $a=\frac{n}{7}$. As $\frac{n}{7}>q^{4}=q^{\frac{m+1}{2}}$, this is impossible.
		
		{\bf Case 3.} Suppose $|C_{a}|=\frac{m}{l}$. Then $n\big|a(q^{\frac{m}{l}}-1)$, or equivalently, $(q^{m}+1)\big|a(q+1)(q^{\frac{m}{l}}-1)$. If $q$ is even, then by Lemma \ref{l3.3}, ${\rm gcd}(q^{m}+1,q^{\frac{m}{l}}-1)=1$ implying $(q^{m}+1)\big|a(q+1)$, that is, $n|a$, which is impossible. If $q$ is odd, then by Lemma \ref{l3.3} again, ${\rm gcd}(q^{m}+1,q^{\frac{m}{l}}-1)=2$ implying $(q^{m}+1)\big|2a(q+1)$; thus $n|2a$ implying $n|a$, which also is impossible.
	\end{proof}
	Next we will find out all coset leaders in the range $[q^{\frac{m-1}{2}}+1, (q-1)q^{\frac{m-1}{2}}]$, and then determine the dimension of  $\mathcal{C}_{(q,n,\delta,1)}$ for $\delta=\ell q^{\frac{m-1}{2}}+1$ with $2\leq \ell\leq q-1$. 
	
	For $m=3$, we have the following result.
	
	\begin{Proposition}\label{p3.2}
		Suppose $n=\frac{q^{3}+1}{q+1}$. Let $a$ be an integer with $q+1\leq a\leq (q-1)q$ and $a\not\equiv0~({\rm mod}~q)$. Denote the $q$-adic expansion of $a$ by $a_{1}q+a_{0}$. Then $a$ is a coset leader except in the following cases:
		
		\noindent {\rm (1)} $\lceil\frac{q}{3}\rceil\leq a_{1}\leq q-2$ and $1\leq a_{0}\leq q-1$,
		
		\noindent {\rm (2)} $1\leq a_{1}\leq \lfloor\frac{q-1}{3}\rfloor$ and $1\leq a_{0}\leq a_{1}$,
		
		\noindent {\rm (3)} $1\leq a_{1}\leq \lfloor\frac{q-3}{3}\rfloor$ and $q-1-2a_{1}\leq a_{0}\leq q-1$,
		
		\noindent {\rm (4)} $a_{1}=\lfloor\frac{q-1}{3}\rfloor$ and $q-2a_{1}\leq a_{0}\leq q-1$.
	\end{Proposition}
	
	\begin{proof}
		 We need to find out all integers $a$ with $q+1\leq a\leq (q-1)q$ and $a\not\equiv0~({\rm mod}~q)$ that are not coset leaders. In other words, for every $1\leq t\leq 2m-1=5$, we need to find out all integers $a$ with $q+1\leq a\leq (q-1)q$ and $a\not\equiv0~({\rm mod}~q)$ satisfying $[aq^{t}]_{n}<a$. By assumption we have $1\leq a_{1}\leq q-2$ and $1\leq a_{0}\leq q-1$. We organize the proof into five cases.
		
		{\bf Case 1}. When $t=1$, we have
		$$aq=a_{1}q^{2}+a_{0}q\equiv a_{1}(q-1)+a_{0}q~({\rm mod}~n).$$
		Denote $M_{1}=a_{1}(q-1)+a_{0}q$. It is easy to see that $M_{1}>0$. Moreover, 
		$$M_{1}\leq (q-2)(q-1)+(q-1)q=2q^{2}-4q+2<2n.$$
		If $0<M_{1}<n$, we have 
		$$[aq]_{n}=M_{1}=a_{1}q+a_{0}q-a_{1}=a_{1}q+a_{0}(q-1)-a_{1}+a_{0}>a_{1}q+a_{0}=a.$$
		Suppose $n<M_{1}<2n$. Then
		$$[aq]_{n}=M_{1}-n=a_{1}(q-1)+a_{0}q-q^{2}+q-1=(a_{1}+a_{0}-q+1)q-a_{1}-1.$$
		If $a_{1}+a_{0}\leq q-1$, then $[aq]_{n}<0$, which is impossible. So $a_{1}+a_{0}\geq q$. Meanwhile, 
		\begin{align*}
			[aq]_{n}<a~\Leftrightarrow &~(a_{0}-q+1)q<a_{1}+a_{0}+1.
		\end{align*}
		The right-hand side inequality above is trivially true. Therefore $[aq]_{n}<a$ if and only if $a$ satisfies the following condition:
		
		\noindent {\bf c1)} $a=a_{1}q+a_{0}$ with $1\leq a_{1}\leq q-2$ and $q-a_{1}\leq a_{0}\leq q-1$.
		
		{\bf Case 2.} When $t=2$, we have
		$$aq^{2}=a_{1}q^{3}+a_{0}q^{2}\equiv -a_{1}+a_{0}(q-1)~({\rm mod}~n).$$
		Denote $M_{2}=-a_{1}+a_{0}(q-1)$. It is clear that $M_{2}>0$ and  $M_{2}\leq -1+(q-1)^{2}=q^{2}-2q<n$. Thus $[aq^{2}]_{n}=M_{2}$, and then
		$$[aq^{2}]_{n}<a~\Leftrightarrow~(a_{0}-a_{1})q<2a_{0}+a_{1}.$$
		Since $3\leq 2a_{0}+a_{1}\leq 2(q-1)+q-2=3q-4$, $[aq^{2}]_{n}<a$ if and only if one of the following holds:
		
		\noindent (i) $a_{0}-a_{1}=2$ and $2q<2a_{0}+a_{1}$,
		
		\noindent (ii) $a_{0}-a_{1}=1$ and $q<2a_{0}+a_{1}$,
		
		\noindent (iii) $a_{0}-a_{1}\leq 0$.
		 
		 \noindent Dealing with these three conditions, we have that $[aq^{2}]_{n}<a$ if and only if $a$ satisfies one of the following three conditions:
		 
		 \noindent {\bf c2)} $a=a_{1}(q+1)+2$ with $\lceil\frac{2q-3}{3}\rceil\leq a_{1}\leq q-3$,
		 
		 \noindent {\bf c3)} $a=a_{1}(q+1)+1$ with $\lceil\frac{q-1}{3}\rceil\leq a_{1}\leq q-2$,
		 
		 \noindent {\bf c4)} $a=a_{1}q+a_{0}$ with $1\leq a_{0}\leq a_{1}\leq q-2$.
		 
		 {\bf Case 3.} When $t=3$, we have $aq^{3}\equiv -a~({\rm mod}~n)$,
		 and then $[aq^{3}]_{n}=n-a$. Hence 
		 $$[aq^{3}]_{n}<a~\Leftrightarrow~(q-1-2a_{1})q<2a_{0}-1.$$
		 As $1\leq 2a_{0}-1\leq 2(q-1)-1=2q-3$, $[aq^{3}]_{n}<a$ if and only if one of the following holds:
		 
		 \noindent (i) $q-1-2a_{1}=1$ and $q<2a_{0}-1$,
		 
		 \noindent (ii) $q-1-2a_{1}\leq 0$.
		 
		 \noindent Direct computation shows that $[aq^{3}]_{n}<a$ if and only if $a$ satisfies one of the following two conditions:
		 
		 \noindent {\bf c5)} $a=\frac{q(q-2)}{2}+a_{0}$ with $q\geq 4$ being even and $\frac{q}{2}+1\leq a_{0}\leq q-1$,
		 
		 \noindent {\bf c6)} $a=a_{1}q+a_{0}$ with $\lceil\frac{q-1}{2}\rceil\leq a_{1}\leq q-2$ and $1\leq a_{0}\leq q-1$.
		 
		 {\bf Case 4.} When $t=4$, we see from Case 1 that $aq^{4}\equiv -aq\equiv -M_{1}~({\rm mod}~n)$,
		 where $M_{1}=a_{1}(q-1)+a_{0}q$ and $0<M_{1}<2n$.
		 
		 Suppose $0<M_{1}<n$. Then $$[aq^{4}]_{n}=n-M_{1}=(q-1-a_{1}-a_{0})q+a_{1}+1.$$
		 If $a_{1}+a_{0}\geq q$, then $[aq^{4}]_{n}<0$, which is impossible, and hence $a_{1}+a_{0}\leq q-1$. Meanwhile,
		 $$[aq^{4}]_{n}<a~\Leftrightarrow~(q-1-2a_{1}-a_{0})q<a_{0}-a_{1}-1.$$
		 Since $-q+2\leq a_{0}-a_{1}-1\leq q-3$, $[aq^{4}]_{n}<a$ if and only if one of the following holds:
		 
		 \noindent (i) $q-1-2a_{1}-a_{0}=0$ and $0<a_{0}-a_{1}-1$,
		 
		 \noindent (ii) $q-1-2a_{1}-a_{0}\leq -1$.
		 
		 \noindent Assume (i), then $a_{0}=q-1-2a_{1}$ and $1\leq a_{1}\leq \lfloor\frac{q-3}{3}\rfloor$. Assume (ii), then ${\rm max}\{1, q-2a_{1}\}\leq a_{0}\leq q-1-a_{1}$ and $1\leq a_{1}\leq q-2$.
		 
		 Suppose $n<M_{1}<2n$. We first note that $M_{1}-n=(a_{1}+a_{0}-q+1)q-a_{1}-1$. It is easy to check that $M_{1}-n\geq 0$ if and only if $a_{1}+a_{0}\geq q$. In addition, 
		 $$[aq^{4}]_{n}=2n-M_{1}=(2q-2-a_{1}-a_{0})q+a_{1}+2.$$
		 It follows that
		 $$[aq^{4}]_{n}<a~\Leftrightarrow~(2q-2-2a_{1}-a_{0})q<a_{0}-a_{1}-2.$$
		 Since $-q+1\leq a_{0}-a_{1}-2\leq q-4$, $[aq^{4}]_{n}<a$ if and only if one of the following holds:
		 
		 \noindent (i) $2q-2-2a_{1}-a_{0}=0$ and $0<a_{0}-a_{1}-2$,
		 
		 \noindent (ii) $2q-2-2a_{1}-a_{0}\leq -1$.
		 
		 \noindent Assume (i), then $a_{0}=2q-2-2a_{1}$ and $\lceil\frac{q-1}{2}\rceil\leq a_{1}\leq \lfloor\frac{2q-5}{3}\rfloor$. Assume (ii), then $\lceil\frac{q}{2}\rceil\leq a_{1}\leq q-1$ and $2q-1-2a_{1}\leq a_{0}\leq q-1$.
		 
		 In conclusion, $[aq^{4}]_{n}<a$ if and only if $a$ satisfies one of the following four conditions:
		 
		 \noindent {\bf c7)} $a=(a_{1}+1)q-2a_{1}-1$ with $1\leq a_{1}\leq \lfloor\frac{q-3}{3}\rfloor$,
		 
		 \noindent {\bf c8)} $a=a_{1}q+a_{0}$ with $1\leq a_{1}\leq q-2$ and ${\rm max}\{1,q-2a_{1}\}\leq a_{0}\leq q-1-a_{1}$,
		 
		 \noindent {\bf c9)} $a=(a_{1}+2)q-2a_{1}-2$ with $\lceil\frac{q-1}{2}\rceil\leq a_{1}\leq \lfloor\frac{2q-5}{3}\rfloor$.
		 
		 \noindent {\bf c10)} $a=a_{1}q+a_{0}$ with $\lceil\frac{q}{2}\rceil\leq a_{1}\leq q-2$ and $2q-1-2a_{1}\leq a_{0}\leq q-1$.
		 
		 {\bf Case 5.} When $t=5$, it follows from Case 2 that $aq^{5}\equiv -aq^{2}\equiv -M_{2}~({\rm mod}~n)$, where $M_{2}=-a_{1}+a_{0}(q-1)$ and $0<M_{2}<n$. Hence $[aq^{5}]_{n}=n-M_{2}=(q-1-a_{0})q+a_{1}+a_{0}+1$. It follows that
		 \begin{align*}
		 	[aq^{5}]_{n}<a~&\Leftrightarrow~(q-1-a_{1}-a_{0})q+a_{1}+1<0\\
		 	&\Leftrightarrow~q-1-a_{1}-a_{0}\leq -1.
		 \end{align*}
		 Therefore $[aq^{5}]_{n}<a$ if and only if $a$ satisfies the following condition:
		 
		 \noindent {\bf c11)} $a=a_{1}q+a_{0}$ with $1\leq a_{1}\leq q-2$ and $q-a_{1}\leq a_{0}\leq q-1$.
		 
		 Summarizing and analysing all the conclusions above, we then arrive at the desired results.
	\end{proof}
	 
	 The following corollary is a direct application of Proposition \ref{p3.2}.
	\begin{Corollary}\label{c3.1}
		Let $n=\frac{q^{3}+1}{q+1}$ and let $\delta_{1}$ be the largest $q$-coset leader modulo $n$. Then
		
	    \noindent {\rm (1)} $\delta_{1}=\frac{q^{2}-2q}{3}$ if $q\equiv 0~({\rm mod}~3)$,
		
		\noindent {\rm (2)}  $\delta_{1}=\frac{q^{2}-3q+2}{3}$ if $q\equiv 1~({\rm mod}~3)$,
		
		\noindent {\rm (3)}  $\delta_{1}=\frac{q^{2}-q+1}{3}$ if $q\equiv 2~({\rm mod}~3)$.
	\end{Corollary}
	\begin{proof}
		Suppose $q\equiv 0~({\rm mod}~3)$. By Proposition \ref{p3.2}, for $q+1\leq a\leq (q-1)q$ with $a\not\equiv0~({\rm mod}~q)$, $a$ is a coset leader except in the following cases:
		
		\noindent (i) $\frac{q}{3}\leq a_{1}\leq q-2$ and $1\leq a_{0}\leq q-1$,
		
		\noindent (ii) $1\leq a_{1}\leq \frac{q-3}{3}$ and $1\leq a_{0}\leq a_{1}$,
		
		\noindent (iii) $1\leq a_{1}\leq \frac{q-3}{3}$ and $q-1-2a_{1}\leq a_{0}\leq q-1$.
		
		\noindent So when $a_{1}\geq \frac{q}{3}$, $a$ is not a coset leader. Let $a_{1}=\frac{q-3}{3}$. Then $q-1-2a_{1}=\frac{q+3}{3}$. We deduce from the cases (ii) and (iii) that $a=\frac{(q-3)q}{3}+a_{0}$ is not a coset leader unless $a_{0}=\frac{q}{3}$. Hence $\delta_{1}=\frac{(q-3)q}{3}+\frac{q}{3}=\frac{q^{2}-2q}{3}$, proving (1). The statements (2) and (3) can be similarly verified.
	\end{proof}
	
	For $n=\frac{q^{3}+1}{q+1}$, the parameters of $\mathcal{C}_{(q,n,\delta_{1},1)}$ and $\mathcal{C}_{(q,n,\delta_{1}+1,0)}$ are presented as follows.
	\begin{Theorem}\label{t3.1}
		Let $n=\frac{q^{3}+1}{q+1}$. Then
		\begin{itemize}[align=left,leftmargin=*]
			\item[{\rm (1)}] if $q\equiv 0~({\rm mod}~3)$, the codes $\mathcal{C}_{(q,n,\delta_{1},1)}$ and $\mathcal{C}_{(q,n,\delta_{1}+1,0)}$ have parameters $[q^{2}-q+1,7,\geq \frac{q^{2}-2q}{3}]$ and $[q^{2}-q+1,6,\geq \frac{2q^{2}-4q}{3}]$, respectively.
			\item[{\rm (2)}] if $q\equiv 1~({\rm mod}~3)$, the codes $\mathcal{C}_{(q,n,\delta_{1},1)}$ and $\mathcal{C}_{(q,n,\delta_{1}+1,0)}$ have parameters $[q^{2}-q+1, 7,\geq \frac{q^{2}-3q+2}{3}]$ and $[q^{2}-q+1, 6,\geq \frac{2q^{2}-6q+4}{3}]$, respectively.
			\item[{\rm (3)}] if $q\equiv 2~({\rm mod}~3)$, the codes $\mathcal{C}_{(q,n,\delta_{1},1)}$ and $\mathcal{C}_{(q,n,\delta_{1}+1,0)}$ have parameters $[q^{2}-q+1,3,\frac{q^{2}-q+1}{3}]$ and $[q^{2}-q+1,2,\frac{2q^{2}-2q+2}{3}]$, respectively.
		\end{itemize}
	\end{Theorem}
	\begin{proof}
		The proof follows from Lemmas \ref{l2.3}-\ref{l2.5}, Proposition \ref{p3.1} and Corollary \ref{c3.1}.
	\end{proof}
	\begin{Example}
		Take $q=5$ in Theorem \ref{t3.1}. Then $n=21$ and $\delta_{1}=7$. By Theorem \ref{t3.1},  the codes $\mathcal{C}_{(q,n,\delta_{1},1)}$ and $\mathcal{C}_{(q,n,\delta_{1}+1,0)}$ have parameters $[21,3,7]$ and $[21,2,14]$, respectively.
	\end{Example}
	
	When $m=3$, the dimension of $\mathcal{C}_{(q,n,\delta,1)}$ is given as follows.
	\begin{Theorem}\label{t3.2}
			Let $n=\frac{q^{3}+1}{q+1}$. For $\delta=\ell q+1$ with $2\leq \ell\leq q-1$, the dimension of the code $\mathcal{C}_{(q,n,\delta,1)}$ is given by
    		\begin{align*}
    			{\rm dim}(\mathcal{C}_{(q,n,\delta,1)})=\begin{cases}
    				n-3\ell(2q-1-3\ell), & {\rm if}~2\leq \ell\leq \lfloor\frac{q-1}{3}\rfloor,\\
    				1, & {\rm if}~\lceil\frac{q}{3}\rceil\leq \ell\leq q-1.
    				\end{cases}
    		\end{align*}
	\end{Theorem}
	\begin{proof}
		When $\lceil\frac{q}{3}\rceil\leq \ell\leq q-1$, we know from Corollary \ref{c3.1} that $\delta>\delta_{1}$, where $\delta_{1}$ is the largest coset leader, which implies that $C_{1}\cup C_{2}\cup\cdots \cup C_{\delta-1}=C_{1}\cup C_{2}\cup\cdots \cup C_{\delta_{1}}$, and hence ${\rm dim}(\mathcal{C}_{(q,n,\delta,1)})=n-|C_{1}\cup C_{2}\cup\cdots \cup C_{\delta_{1}}|=|C_{0}|=1$. 
		
		For $2\leq \ell \leq \lfloor\frac{q-1}{3}\rfloor$, denote by $N_{\ell}$ the set consisting of all integers $a$ with $q+1\leq a\leq (q-1)q$ and $a\not\equiv0~({\rm mod}~q)$ that are not coset leaders. By Proposition \ref{p3.2}, $N_{\ell}=N_{1,\ell}\cup N_{2,\ell}$, where 
		\begin{align*}
			N_{1,\ell}&=\big\{a=a_{1}q+a_{0}~\big|~1\leq a_{1}\leq \ell-1,~1\leq a_{0}\leq a_{1}\big\},\\
			N_{2,\ell}&=\big\{a=a_{1}q+a_{0}~\big|~1\leq a_{1}\leq \ell-1,~q-1-2a_{1}\leq a_{0}\leq q-1\big\}.
		\end{align*}
		It is easy to check that $N_{1,\ell}\cap N_{2,\ell}=\emptyset$. Then $|N_{\ell}|=|N_{1,\ell}|+|N_{2,\ell}|=\frac{(3\ell+2)(\ell-1)}{2}$. On the one hand, the number of integers $a$ with $1\leq a\leq \delta-1$ that are multiples of $q$ is $\lfloor\frac{\delta-1}{q}\rfloor$. On the other hand, it follows from Lemma \ref{l3.2} that the number of integers $a$ with $1\leq a\leq \delta-1$ and $a\not\equiv0~({\rm mod}~q)$ that are not coset leaders is $|N_{\ell}|+1$. Therefore $C_{1}\cup C_{2}\cup\cdots \cup C_{\delta-1}$ is the union of $\delta-1-\lfloor\frac{\delta-1}{q}\rfloor-|N_{\ell}|-1=\ell(q-1)-|N_{\ell}|-1$ distinct $q$-cosets. By Proposition \ref{p3.1}, all these $q$-cosets have cardinality $6$. Hence 
		$${\rm dim}(\mathcal{C}_{(q,n,\delta,1)})=n-6(\ell(q-1)-|N_{\ell}|-1)=n-3\ell(2q-1-3\ell).$$
        The proof is then completed.
		\end{proof}
	\begin{Example}{\rm
		Take $q=8$ in Theorem \ref{t3.2}. Then $n=57$ and $2\leq \ell\leq q-1=6$. Let $\ell=2$, then $\delta=17$. By Theorem \ref{t3.2}, the codes $\mathcal{C}_{(q,n,\delta,1)}$ and $\mathcal{C}_{(q,n,\delta+1,0)}$ have parameters $[57,3,\geq 17]$ and $[57,2,\geq 34]$, respectively.}
	\end{Example}
	
	For $m=5$, using similar methods as $m=3$ we obtain the following result.
	\begin{Proposition}\label{p3.3}
		Suppose $n=\frac{q^{5}+1}{q+1}$. Let $a$ be an integer with $q^{2}+1\leq a\leq (q-1)q^{2}$ and $a\not\equiv 0~({\rm mod}~q)$. Denote the $q$-adic expansion of $a$ by $a_{2}q^{2}+a_{1}q+a_{0}$. 
		\begin{enumerate}[align=left,leftmargin=*]
			\setlength{\parindent}{0.6em}
			\renewcommand{\labelenumi}{{\rm (\theenumi)}}
			\renewcommand{\labelenumii}{{\rm (\roman*)}}
			\item When $q=3$, $a$ is a coset leader except that $a\in \{11,13,14,16,17\}$.
			\item When $q=4$, $a$ is a coset leader except that $a\in\{19,25,26,27,29,35,37,38,39,42,43,\\
			\indent 45,46,47\}$.
			\item When $q\geq 5$, $a$ is a coset leader except in the following cases:
			\begin{enumerate}
				\item  $a\in \{q^{2}+q+1,(q-2)q^{2}+q-1\}$,
				\item  $a=\frac{q^{3}+q^{2}+q}{2}-1$, where $q\geq 5$ is even,
				\item $a=a_{2}q^{2}+(q-1-a_{2})q+2a_{2}+1$, where $1\leq a_{2}\leq {\rm min}\{\lfloor\frac{q}{2}\rfloor-1,q-3\}$,
				\item $a=a_{2}q^{2}+(q-1-a_{2})q+a_{0}$, where $1\leq a_{2}\leq q-3$ and $1\leq a_{0}\leq {\rm min}\{2a_{2},q-1\}$,
				\item $a=a_{2}q^{2}+(q-a_{2})q+a_{0}$, where $\lceil\frac{q+1}{2}\rceil\leq a_{2}\leq q-3$ and $1\leq a_{0}\leq 2a_{2}-q$,
				\item $a=a_{2}q^{2}+(a_{1}+1)q-a_{1}$, where $1\leq a_{2}\leq q-3$ and $q-1-a_{2}\leq a_{1}\leq q-1$,
				\item $a=a_{2}q^{2}+(a_{2}+1)q-a_{2}-1$, where $\lceil\frac{q+1}{2}\rceil\leq a_{2}\leq q-3$,
				\item $a=a_{2}q^{2}+a_{2}q-a_{2}$, where $2\leq a_{2}\leq q-3$,
				\item $a=a_{2}q^{2}+(a_{1}+1)q-a_{1}-1$, where $2\leq a_{2}\leq q-3$ and $0\leq a_{1}\leq a_{2}-2$,
				\item $a=(q-2)q^{2}+a_{1}q+a_{0}$, where $1\leq a_{1}, a_{0}\leq q-1$.
			\end{enumerate}
		\end{enumerate}
	\end{Proposition}
	
	Combining Lemma \ref{l3.2}, Propositions \ref{p3.1} and \ref{p3.3}, we have the following theorem.
	\begin{Theorem}\label{t3.3}
		Let $n=\frac{q^{5}+1}{q+1}$. For $\delta=\ell q^{2}+1$ with $2\leq \ell\leq q-1$, the dimension of the code $\mathcal{C}_{(q,n,\delta,1)}$ is given as follows.
		
		\noindent {\rm (1)} When $q=3$, we have $\delta=2q^{2}+1$ and ${\rm dim}(\mathcal{C}_{(q,n,\delta,1)})=1$.
		
		\noindent {\rm (2)} When $q=4$, 
		\begin{align*}
			{\rm dim}(\mathcal{C}_{(q,n,\delta,1)})=\begin{cases}
				25, & {\rm if}~\ell=2,\\
				3, & {\rm if}~\ell=3.
			\end{cases}
		\end{align*}
		\noindent {\rm (3)} When $q\geq 5$, 
		\begin{align*}
			{\rm dim}(\mathcal{C}_{(q,n,\delta,1)})=\begin{cases}
				n-10\big(\ell(q-1)q-2\ell^{2}+\ell\big), & {\rm if}~2\leq \ell \leq \lceil\frac{q-1}{2}\rceil,\\
				n-10\big((q^{2}-q-2\lfloor\frac{q}{2}\rfloor)\lceil\frac{q+1}{2}\rceil-\lceil\frac{q}{2}\rceil+1\big), & {\rm if}~\ell=\lceil\frac{q+1}{2}\rceil,\\
				n-10(\ell(q-1)q-2\ell^{2}+3\ell-q), & {\rm if}~\lceil\frac{q+3}{2}\rceil\leq \ell \leq q-2,\\
				q^{4}-11q^{3}+51q^{2}-131q+161, & {\rm if}~\ell=q-1.
			\end{cases}
		\end{align*}
	\end{Theorem}
	
	\begin{Example}{\rm
		Take $q=5$ in Theorem \ref{t3.3}. Then $n=521$ and $2\leq \ell \leq q-1=4$. Let $\ell=4$, then $\delta =101$. By Theorem \ref{t3.3}, the codes $\mathcal{C}_{(q,n,\delta,1)}$ and $\mathcal{C}_{(q,n,\delta+1,0)}$ have parameters $[521,31, \geq 101]$ and $[521,30,\geq 202]$, respectively.}
	\end{Example}
	
	Below we consider the case that $m>5$. Set $h=\frac{m-1}{2}$, then $h>2$. When $h$ is odd, we have the following result. Since the proof is long and technical, we defer it to Appendix.
		\begin{Proposition}\label{p3.4}
		Let $n=\frac{q^m+1}{q+1}$, where $m>5$ is odd. Set $h=\frac{m-1}{2}$. Let $a$ be an integer with $q^{h}+1\leq a\leq (q-1)q^{h}$ and $a\not\equiv0~({\rm mod}~q)$. Denote the $q$-adic expansion of $a$ by $\sum_{i=0}^{h}a_{i}q^{i}$. Assume that $h>2$ is odd, then $a$ is a coset leader except in the following cases:
		
		\noindent {\rm (1)} $a=(a_{h}+1)\frac{q^{h+1}+q}{q+1}-q+a_{0}$, where $1\leq a_{h}\leq q-2$ and ${\rm max}\{1,q-1-2a_{h}\}\leq a_{0}\leq q-1$.
		
		\noindent {\rm (2)} $a=(a_{h}+1)\frac{q^{h+1}+q}{q+1}-2q+a_{0}$, where $\lceil\frac{q+1}{2}\rceil \leq a_{h}\leq q-2$ and $2q-2a_{h}\leq a_{0}\leq q-1$.
		
		\noindent {\rm (3)} $a=a_{h}q^{h}+(a_{h}+1)\frac{q^{h}+1}{q+1}$, where $\lceil\frac{q-1}{2}\rceil\leq a_{h}\leq q-2$,
		
		\noindent {\rm (4)} $a=a_{h}q^{h}+(a_{h-1}+1)\frac{q^{h}+1}{q+1}$, where $0\leq a_{h-1}< a_{h}\leq q-2$.
		
		\noindent {\rm (5)} $a=a_{h}q^{h}+(a_{h-1}+1)\frac{q^{h}+1}{q+1}-1$, where $1\leq a_{h}\leq q-2$ and $q-a_{h}\leq a_{h-1}\leq q-1$.
		\end{Proposition}
		
		The case that $h$ is even can be treated similarly as the case that $h$ is odd and the proof of the following proposition is thus omitted. 
		\begin{Proposition}
		Let $n=\frac{q^m+1}{q+1}$, where $m>5$ is odd. Set $h=\frac{m-1}{2}$. Let $a$ be an integer with $q^{h}+1\leq a\leq (q-1)q^{h}$ and $a\not\equiv0~({\rm mod}~q)$. Denote the $q$-adic expansion of $a$ by $\sum_{i=0}^{h}a_{i}q^{i}$. Assume that $h>2$ is even, then $a$ is a coset leader except in the following cases:
			
			\noindent {\rm (1)} $a=(a_{h}+1)\frac{q^{h+1}-q}{q+1}+a_{0}$, where $1\leq a_{h}\leq q-2$ and $1\leq a_{0}\leq {\rm min}\{2a_{h}+1,q-1\}$,
			
			\noindent {\rm (2)} $a=(a_{h}+1)\frac{q^{h+1}-q}{q+1}+q+a_{0}$, where $\lceil\frac{q+1}{2}\rceil\leq a_{h}\leq q-2$ and $1\leq a_{0}\leq 2a_{h}-q$,
			
			\noindent {\rm (3)} $a=a_{h}q^{h}+(a_{h}+1)\frac{q^{h}-1}{q+1}$, where $\lceil\frac{q}{2}\rceil\leq a_{h}\leq q-2$,
			
			\noindent {\rm (4)} $a=a_{h}q^{h}+(a_{h-1}+1)\frac{q^{h}-1}{q+1}$, where $0\leq a_{h-1}< a_{h}\leq q-2$,
			
			\noindent {\rm (5)} $a=a_{h}q^{h}+(a_{h-1}+1)\frac{q^{h}-1}{q+1}+1$, where $1\leq a_{h}\leq q-2$ and $q-1-a_{h}\leq a_{h-1}\leq q-1$.
		\end{Proposition}

		\begin{Theorem}\label{t3.4}
			Let $n=\frac{q^{m}+1}{q+1}$, where $m>5$ is an odd integer. For $\delta=\ell q^{\frac{m-1}{2}}+1$ with $2\leq \ell\leq q-1$, the dimension of the code $\mathcal{C}_{(q,n,\delta,1)}$ is given as follows:
			\begin{align*}
				{\rm dim}(\mathcal{C}_{(q,n,\delta,1)})=\begin{cases}
					n-2m\big(\ell(q-1)q^{\frac{m-3}{2}}-2\ell^{2}+\ell \big), & {\rm if}~2\leq \ell\leq \lceil\frac{q-1}{2}\rceil,\\
					n-2m\big(\ell(q-1)q^{\frac{m-3}{2}}-2\lceil\frac{q+1}{2}\rceil \lceil\frac{q-1}{2}\rceil-\lceil\frac{q-1}{2}\rceil\big), & {\rm if}~\ell=\lceil\frac{q+1}{2}\rceil,\\
					n-2m\big(\ell(q-1)q^{\frac{m-3}{2}}-2\ell^{2}+3\ell-q\big), & {\rm if}~\lceil\frac{q+3}{2}\rceil \leq \ell\leq q-1.
				\end{cases}
			\end{align*}
		\end{Theorem}
		
		\begin{proof}
			Set $h=\frac{m-1}{2}$. We first assume that $h$ is odd. For $2\leq \ell \leq q-1$, denote by $N_{\ell}$ the set consisting of all integers $a$ with $q^{h}+1\leq a\leq \delta-1=\ell q^{h}$ and $a\not\equiv 0~({\rm mod}~q)$ that are not coset leaders. It follows from Proposition \ref{p3.4} that $N_{\ell}=N_{1,\ell}\cup N_{2,\ell}\cup N_{3,\ell}\cup N_{4,\ell}\cup N_{5,\ell}$, where
			{\small \begin{align*}
				N_{1,\ell}&=\big\{a=(a_{h}+1)\frac{q^{h+1}+q}{q+1}-q+a_{0}~\big|~1\leq a_{h}\leq \ell-1,~{\rm max}\{1,q-1-2a_{h}\}\leq a_{0}\leq q-1\big\},\\
				N_{2,\ell}&=\big\{a=(a_{h}+1)\frac{q^{h+1}+q}{q+1}-2q+a_{0}~\big|~\lceil\frac{q+1}{2}\rceil \leq a_{h}\leq \ell-1,~2q-2a_{h}\leq a_{0}\leq q-1\big\},\\
				N_{3,\ell}&=\big\{a=a_{h}q^{h}+(a_{h}+1)\frac{q^{h}+1}{q+1}~\big|~\lceil\frac{q-1}{2}\rceil\leq a_{h}\leq \ell-1\big\},\\
				N_{4,\ell}&=\big\{a=a_{h}q^{h}+(a_{h-1}+1)\frac{q^{h}+1}{q+1}~\big|~0\leq a_{h-1}< a_{h}\leq \ell-1\big\},\\
				N_{5,\ell}&=\big\{a=a_{h}q^{h}+(a_{h-1}+1)\frac{q^{h}+1}{q+1}-1~\big|~1\leq a_{h}\leq \ell-1,~q-a_{h}\leq a_{h-1}\leq q-1\big\}.
			\end{align*}}
		 For $1\leq i<j\leq 5$, one can check that $N_{i,\ell}\cap N_{j,\ell}=\emptyset$ except that 
			$$N_{1,\ell}\cap N_{3,\ell}=\big\{a=\frac{q^{h+1}+1}{2}~\big|~q\geq 3~{\rm is~odd},~\frac{q+1}{2}\leq \ell \leq q-1\big\},$$
			and
			$$N_{1,\ell}\cap N_{4,\ell}=\big\{a=(a_{h}+1)\frac{q^{h+1}-1}{q+1}+1~\big|~\lceil\frac{q}{2}\rceil\leq a_{h}\leq \ell-1\big\}.$$
			Applying the inclusion-exclusion principle we obtain
			\begin{align*}
				|N_{\ell}|&=|N_{1,\ell}\cup N_{2,\ell}\cup N_{3,\ell}\cup N_{4,\ell}\cup N_{5,\ell}|\\
				&=|N_{1,\ell}|+|N_{2,\ell}|+|N_{3,\ell}|+|N_{4,\ell}|+|N_{5,\ell}|-|N_{1,\ell}\cap N_{3,\ell}|-|N_{1,\ell}\cap N_{4,\ell}|\\
				&=\begin{cases}
					2\ell^{2}-\ell-1, & {\rm if}~2\leq \ell\leq \lceil\frac{q-1}{2}\rceil,\\
					2\lceil\frac{q+1}{2}\rceil\lceil\frac{q-1}{2}\rceil+\lceil\frac{q-1}{2}\rceil-1, & {\rm if}~\ell=\lceil\frac{q+1}{2}\rceil,\\
					2\ell^{2}-3\ell+q-1, & {\rm if}~\lceil\frac{q+3}{2}\rceil \leq \ell\leq q-1.
					\end{cases}
			\end{align*}
			By Lemma \ref{l3.2}, Propositions \ref{p3.1} and \ref{p3.4}, we have 
			$${\rm dim}(\mathcal{C}_{(q,n,\delta,1)})=n-2m\big(\delta-1-\lfloor\frac{\delta-1}{q}\rfloor-|N_{\ell}|-1\big)=n-2m\big(\ell(q-1)q^{\frac{m-3}{2}}-|N_{\ell}|-1\big).$$
			The desired conclusion then follows.
			
			When $h$ is even, the proof is very similar as above, and thus omitted here.
	    \end{proof}
	  
	  \begin{Example}{\rm
	  		Take $q=3$ and $m=7$. Then $n=547$ and $\delta=55$. By Theorem \ref{t3.4}, the codes $\mathcal{C}_{(q,n,\delta,1)}$ and $\mathcal{C}_{(q,n,\delta+1,0)}$ have parameters $[547,113,\geq 55]$ and $[547,112,\geq 110]$, respectively.}
	 \end{Example}
	 
	 \begin{Remark}{\rm
	 	Using the similar methods as above, we also can find out all coset leaders in the range $\big[(q-1)q^{\frac{m-1}{2}}+1, q^{\frac{m+1}{2}}\big]$, where $q\geq 2$.}
	 \end{Remark}
	 
	\section{Parameters of some binary BCH codes with length $\frac{2^m+1}{3}$}
	In this section, we study binary BCH codes of length $n=\frac{2^m+1}{3}$ with small dimensions. We will describe the first two largest $2$-coset leaders $\delta_{1}$ and $\delta_{2}$ modulo $n$, and then investigate the parameters of the BCH codes related to $\delta_{1}$ and $\delta_{2}$.
	 
	To determine the values of $\delta_{1}$ and $\delta_{2}$, we need the following lemma.
	
	\begin{Lemma}\label{l4.1}
		Suppose that $m$ is odd. Let $s$ be an integer with $0<s \leq 2^{m-1}-1$ and $s\not\equiv 0~({\rm mod}~2)$. Assume that in the $2$-adic expansion of $s$, there are no three consecutive $0$'s nor three consecutive $1$'s. Let $l$ be the number of pairs of $0$'s and $1$'s that appear together in the $2$-adic expansion of $s$, then $l$ is odd.
	\end{Lemma}
	
	\begin{proof}
		The $2$-adic expansion of the integer $s$ has the following form:
		 $$s=(s_{m-1}=0, s_{m-2},\cdots,s_{1},s_{0}=1),$$
		 where $s_{j}\in\{0,1\}$ for $1\leq j\leq m-2$. In the $2$-adic expansion of $s$, if all the pairs of $0$'s and $1$'s are treated simply as $0$'s and $1$'s, respectively, then we obtain a sequence $(0,1,0,1,\cdots,0,1)$, which is of even length; thus $m-l$ is even implying that $l$ is odd.
	\end{proof}
	
	\begin{Proposition}\label{p4.1}
		Let $q=2$ and $n=\frac{2^{m}+1}{3}$, where $m$ is an odd integer. Suppose $m\equiv 0~({\rm mod}~3)$, then:
		
		\noindent {\rm (1)} if $m=3$, then $\delta_{1}=1$ with $|C_{\delta_{1}}|=2$.
		
		\noindent {\rm (2)} if $m>3$, then $\delta_{1}=\frac{2^m+1}{9}$ and $\delta_{2}=\frac{2^{m-1}-31}{9}$; moreover, $|C_{\delta_{1}}|=2$ and $|C_{\delta_{2}}|=2m$.
	\end{Proposition}
	\begin{proof}
		We first verify the value of $\delta_{1}$. Since the $2$-coset modulo $n$ containing $\frac{2^{m}+1}{9}$ is $\{\frac{2^{m}+1}{9},\frac{2^{m+1}+2}{9}\}$, $\frac{2^m+1}{9}$ is a $2$-coset leader modulo $n$ and $|C_{\frac{2^m+1}{9}}|=2$. In order to prove $\delta_{1}=\frac{2^m+1}{9}$, we need to show that for every integer $i$ with $\frac{2^m+1}{9}<i\leq n-1$, $i$ is not a $2$-coset leader modulo $n$, which is equivalent to prove that $3i$ is not a $2$-coset leader modulo $3n$. In other words, we should prove that for $\frac{2^m+1}{9}<i\leq n-1$, there exists an integer $t$ with $1\leq t\leq 2m-1$ such that $[3i\cdot 2^{t}]_{3n}<3i$. Note that $C_{i}=C_{-i}$ for $0\leq i\leq n-1$; thus we only need consider $i$ with $\frac{2^{m}+1}{9}<i\leq \frac{n-1}{2}=\frac{2^{m-1}-1}{3}$.
		
		Let $i$ be an integer with $\frac{2^{m}+1}{9}<i\leq \frac{2^{m-1}-1}{3}$ and $i\not\equiv 0~({\rm mod}~2)$, then $\frac{2^{m}+1}{3}=(0,1,0,1,\cdots,0,1,1)<3i\leq 2^{m-1}-1=(0,1,1,\cdots,1)$. The $2$-adic expansion of $3i$ can be written as 
		$$3i=(0,i_{m-2}=1,i_{m-3},\cdots,i_{s+1},i_{s}=1,i_{s-1}=1,i_{s-2},\cdots,i_{1},i_{0}=1),$$
		where $3\leq s\leq m-2$, and $i_{j}\in\{0,1\}$ for $1\leq j\leq s-2$ and $s+1\leq j\leq m-3$. It follows that
		\begin{align*}
			[3i\cdot 2^{2m-1-s}]_{3n}&=3n-[3i\cdot 2^{m-1-s}]_{3n}\\
			&=(1,1,\cdots,1,3)-(1,1,i_{s-2},\cdots,i_{1},1,0,-i_{m-2},\cdots,-i_{s+1})\\
			&=(0,0,1-i_{s-2},\cdots,1-i_{1},0,1,1+i_{m-2},\cdots,1+i_{s+2},3+i_{s+1})\\
			&\leq (0,0,1-i_{s-2},\cdots,1-i_{1},0,1,\overset{m-3-s}{\overbrace{2,2,\cdots,2}},4)\\
			&=(0,0,1-i_{s-2},\cdots,1-i_{1},1,1,\overset{m-2-s}{\overbrace{0,0,\cdots,0}})<3i.
		\end{align*}
		Thus for $\frac{2^{m}+1}{9}<i\leq \frac{2^{m-1}-1}{3}$, $3i$ is not a $2$-coset leader modulo $3n$, and hence $i$ is not a $2$-coset leader modulo $n$, proving $\delta_{1}=\frac{2^m+1}{9}$.
		
		To show that $\delta_{2}=\frac{2^{m-1}-31}{9}$ and $|C_{\delta_{2}}|=2m$, we need to prove the following two statements:
		
		\noindent {\bf I.} For $\frac{2^{m-1}-31}{9}< i<\delta_{1}=\frac{2^{m}+1}{9}$ with $i\not\equiv 0~({\rm mod}~2)$, $i$ is not a $2$-coset leader modulo $n$.
		
		\noindent {\bf II.} $i=\frac{2^{m-1}-31}{9}$ is a $2$-coset leader modulo $n$ and $|C_{i}|=2m$.
		
		We first prove the statement {\bf I}, and we divide the proof into two cases: $\frac{2^{m-1}-13}{9}\leq i\leq \frac{2^{m-2}-2}{3}$ and $\frac{2^{m-2}+1}{3}\leq i<\frac{2^{m}+1}{9}$.
		
		{\bf Case I.1.} When $\frac{2^{m-1}-13}{9}\leq i\leq \frac{2^{m-2}-2}{3}$, $\frac{2^{m-1}-13}{3}=(0,0,\overset{m-5}{\overbrace{1,0,\cdots,1,0}},0,0,1)\leq 3i<2^{m-2}-1=(0,0,1,1,\cdots,1)$. The $2$-adic expansion of $3i$ is of the following form:
		\begin{equation}\label{s3}
			3i=(0,0,i_{m-3}=1,i_{m-4},\cdots,i_{1},i_{0}=1),
		\end{equation}
		where $i_{j}\in\{0,1\}$ for $1\leq j\leq m-4$. 
		
		{\bf Subcase I.1.1.} Suppose that (\ref{s3}) has at least three consecutive $0$'s, that is,
		$$3i=(0,0,1,i_{m-4},\cdots,i_{s}=0,i_{s-1}=0,i_{s-2}=0,\cdots,i_{1},1),$$
		where $3\leq s\leq m-4$. Then
		$$[3i\cdot 2^{m-1-s}]=(0,0,0,i_{s-3},\cdots,i_{1},1,0,0,-1,-i_{m-4},\cdots,-i_{s+1})<3i.$$
		
		{\bf Subcase I.1.2.} Suppose that (\ref{s3}) has at least three consecutive $1$'s, namely,
		$$3i=(0,0,i_{m-3},\cdots,i_{s}=1,i_{s-1}=1,i_{s-2}=1,\cdots,i_{0}),$$
		where $2\leq s\leq m-3$. Then
		\begin{align*}
			[3i\cdot 2^{2m-1-s}]_{3n}&=3n-[3i\cdot 2^{m-1-s}]_{3n}\\
			&=(1,1,\cdots,1,3)-(1,1,1,i_{s-3},\cdots,i_{0},0,0,-i_{m-3},\cdots,-i_{s+1})\\
			&=(0,0,0,1-i_{s-3},\cdots,1-i_{0},1,1,1+i_{m-3},\cdots,1+i_{s+2},3+i_{s+1})\\
			&\leq (0,0,0,\overset{s-2}{\overbrace{1,1,\cdots,1}},1,1,\overset{m-4-s}{\overbrace{2,2,\cdots,2}},4)\\
			&=(0,0,1,\overset{s-1}{\overbrace{0,0,\cdots,0}},1,\overset{m-3-s}{\overbrace{0,0,\cdots,0}})\\
			&\leq (0,0,1,0,1,\overset{m-5}{\overbrace{0,0,\cdots,0}})<\frac{2^{m-1}-13}{3}\leq 3i.
		\end{align*}
		
		{\bf Subcase I.1.3.} Suppose that (\ref{s3}) has no three consecutive $0$'s nor three consecutive $1$'s. Then $3i$ can have pairs of $0$'s, pairs of $1$'s, $0$'s and $1$'s in its $2$-adic expansion. Let $l$ be the number of pairs of $0$'s and pairs of $1$'s together in (\ref{s3}). By Lemma \ref{l4.1}, $l$ is odd. 
		
		If $l=1$, then $3i=(0,0,1,0,1,0,\cdots,1,0,1)$, which is impossible. Otherwise, $3i=\frac{2^{m-1}-1}{3}$ implying $2^{m-1}\equiv 1~({\rm mod}~9)$; thus $m\equiv 1~({\rm mod}~6)$ as ${\rm ord}_{9}(2)=6$, and hence $m\equiv 1~({\rm mod}~3)$, a contradiction to the assumption that $m\equiv 0~({\rm mod}~3)$.
		
		Assume $l\geq 3$. Let $l_{0}$ and $l_{1}$ be the number of pairs of $0$'s and pairs of $1$'s in (\ref{s3}), respectively. As $l_{0}+l_{1}=l$ and $l$ is odd, $l_{0}\neq l_{1}$ implying $l_{0}<l_{1}$ or $l_{0}>l_{1}$.
		
		Suppose $l_{0}<l_{1}$. The $2$-adic expansion of $3i$ must contain $(1,1,\overset{\ell}{\overbrace{0,1,\cdots,0,1}},0,1,1)$, where $0\leq \ell\leq m-7$. More precisely, 
		$$3i=(0,0,i_{m-3},\cdots,i_{s}=1,1,\overset{\ell}{\overbrace{0,1,\cdots,0,1}},0,1,i_{s-\ell-4}=1,\cdots,i_{0}),$$
		where $\ell+4\leq s\leq m-3$. If $\ell=m-7$, then $3i=(0,0,1,1,\overset{m-7}{\overbrace{0,1,\cdots,0,1}},0,1,1)$,
		which is impossible. Otherwise, $3i=3(2^{m-4}+1)+\frac{2^{m-4}-8}{3}$ implying $2^{m-4}\equiv -1~({\rm mod}~9)$; so $2^{m-1}\equiv 1~({\rm mod}~9)$, and then $m\equiv 1~({\rm mod}~3)$, contradicting $m\equiv 0~({\rm mod}~3)$. If $\ell=m-9$, then
		$$3i=(0,0,1,0,1,1,\overset{m-9}{\overbrace{0,1,\cdots,0,1}},0,1,1)~{\rm or}~3i=(0,0,1,1,\overset{m-9}{\overbrace{0,1,\cdots,0,1}},0,1,1,0,1).$$
		One can check that these two forms of $3i$ are all imposiible as they cannot be divided by $3$. If $\ell<m-9$, then
		\begin{align*}
			&[3i\cdot 2^{2m-1-s}]_{3n}=3n-[3i\cdot 2^{m-1-s}]_{3n}\\
			=&(1,1,\cdots,1,3)-(1,1,\overset{\ell}{\overbrace{0,1,\cdots,0,1}},0,1,1,i_{s-\ell-5}\cdots,i_{0},0,0,-i_{m-3},\cdots,-i_{s+1})\\
			=&(0,0,\overset{\ell}{\overbrace{1,0,\cdots,1,0}},1,0,0,1-i_{s-\ell-5},\cdots,1-i_{0},1,1,1+i_{m-3},\cdots,1+i_{s+2},3+i_{s+1})\\
			\leq &(0,0,\overset{\ell}{\overbrace{1,0,\cdots,1,0}},1,0,0,\overset{s-\ell-4}{\overbrace{1,1,\cdots,1}},1,1,\overset{m-4-s}{\overbrace{2,2,\cdots,2}},4)\\
			=&(0,0,\overset{\ell}{\overbrace{1,0,\cdots,1,0}},1,0,1,\overset{s-\ell-3}{\overbrace{0,0,\cdots,0}},1,\overset{m-3-s}{\overbrace{0,0,\cdots,0}})\\
			\leq &(0,0,\overset{m-11}{\overbrace{1,0,\cdots,1,0}},1,0,1,0,1,0,0,0,0)<\frac{2^{m-1}-13}{3}\leq 3i.
		\end{align*}
		
		Suppose $l_{0}>l_{1}$. Then 
		$$3i=(0,0,i_{m-3},\cdots,i_{s+1},i_{s}=0,0,\overset{\ell}{\overbrace{1,0,\cdots,1,0}},1)~{\rm with}~s\leq m-5,~{\rm or}$$
		$$3i=(0,0,i_{m-3},\cdots,i_{s+1},i_{s}=0,0,\overset{\ell}{\overbrace{1,0,\cdots,1,0}},1,0,0,i_{s-t-5},\cdots,i_{0})~{\rm with}~t+5\leq s\leq m-1,$$
		where $0\leq \ell\leq m-7$. If $3i$ is of the first form, then
		$$[3i\cdot 2^{m-1-s}]_{3n}=(0,0,\overset{\ell}{\overbrace{1,0,\cdots,1,0}},1,0,0,-i_{m-3},\cdots,-i_{s+1})<\frac{2^{m-1}-13}{3}\leq 3i.$$ 
		Assume that $3i$ is of the second form. If $\ell=m-7$, then $3i=(0,0,\overset{m-7}{\overbrace{1,0,\cdots,1,0}},1,0,0,1,1)=\frac{2^{m-1}-7}{3}$, which is imposiible as $3\nmid \frac{2^{m-1}-7}{3}$. If $\ell\leq m-9$, then 
		\begin{align*}
			[3i\cdot 2^{m-1-s}]_{3n}&=(0,0,\overset{\ell}{\overbrace{1,0,\cdots,1,0}},1,0,0,i_{s-t-5},\cdots,i_{0},0,0,-i_{m-3},\cdots,-i_{s+1})\\
			&<\frac{2^{m-1}-13}{3}\leq 3i.
		\end{align*}
		
		{\bf Case I.2.} When $\frac{2^{m-2}+1}{3}\leq i<\frac{2^{m}+1}{9}$, $2^{m-2}+1=(0,1,0,0,\cdots,0,1)\leq 3i< \frac{2^{m}-2}{3}=(0,1,0,1,0,\cdots,1,0)$. The $2$-adic expansion of $3i$ is of the following form:
		\begin{equation}\label{s2}
			3i=(0,i_{m-2}=1,i_{m-3},\cdots,i_{1},i_{0}=1),
		\end{equation}
		where $i_{j}\in\{0,1\}$ for $1\leq j\leq m-3$.  If $(0,0)$ appears in (\ref{s2}), then there exists an integer $t$ with $1\leq t\leq m-1$ such that $[3i\cdot 2^{t}]_{3n}<3i$. Suppose that (\ref{s2}) has no two consecutive $0$'s, then $3i=(0,1,0,1,0,\cdots,1,0)$, contradicting $i_{0}=1$. 
		
		We conclude that for every integer $i$ with $\frac{2^{m-1}-31}{9}< i<\frac{2^{m}+1}{9}$ and $i\not\equiv 0~({\rm mod}~2)$, $3i$ is not a $2$-coset leader modulo $3n$, and hence $i$ is not a $2$-coset leader modulo $n$, proving {\bf I}.
		
		To prove the statement {\bf II}, that is, to prove that $i=\frac{2^{m-1}-31}{9}$ is a $2$-coset leader modulo $n$ and $|C_{i}|=2m$, it suffices to show that $[i\cdot 2^{t}]_{n}>i$ for $1\leq t\leq 2m-1$, which is equivalent to show that $[3i\cdot 2^{t}]_{3n}>3i$ for $1\leq t\leq 2m-1$. Note that $3i=\frac{2^{m-1}-31}{3}=(0,0,\overset{m-9}{\overbrace{1,0,\cdots,1,0}},1,0,0,1,0,1,1)$. We organize the proof into ten cases.
		
		{\bf Case II.1.} When $t=1$, $[3i\cdot 2]_{3n}=(0,\overset{m-9}{\overbrace{1,0,\cdots,1,0}},1,0,0,1,0,1,1,0)>3i$.
		
		{\bf Case II.2.} When $t=2$, $[3i\cdot 2^{2}]_{3n}=(\overset{m-9}{\overbrace{1,0,\cdots,1,0}},1,0,0,1,0,1,1,0,0)>3i$.
		
		{\bf Case II.3.} Suppose $3\leq t\leq m-7$. If $t$ is odd, then
		\begin{align*}
			[3i\cdot 2^{t}]_{3n}&=(0,\overset{m-8-t}{\overbrace{1,0,\cdots,1,0}},1,0,0,1,0,1,1,0,0,-1,\overset{t-3}{\overbrace{0,-1,\cdots,0,-1}})\\
			&>(0,\overset{m-8-t}{\overbrace{1,0,\cdots,1,0}},1,0,0,1,0,1, 0,1,1,0,\overset{t-3}{\overbrace{1,0,\cdots,1,0}})>3i.
		\end{align*}
		If $t$ is even, then
		\begin{align*}
			[3i\cdot 2^{t}]_{3n}&=(\overset{m-7-t}{\overbrace{1,0,\cdots,1,0}},1,0,0,1,0,1,1,0,0,\overset{t-2}{\overbrace{-1,0,\cdots,-1,0}})\\
			&>(\overset{m-7-t}{\overbrace{1,0,\cdots,1,0}},1,0,0,1,0,1,0,1,1,\overset{t-2}{\overbrace{0,1,\cdots,0,1}})>3i.
		\end{align*}
		
		{\bf Case II.4.} When $t=m-6$, we have 
		\begin{align*}
			[3i\cdot 2^{m-6}]_{3n}&=(0,0,1,0,1,1, 0,0,\overset{m-9}{\overbrace{-1,0,\cdots,-1,0}},-1)\\
			&>(0,0,1,0,1,0, 1,1,\overset{m-9}{\overbrace{0,1,\cdots,0,1}},0)>3i.
		\end{align*}
		
		{\bf Case II.5.} When $m-5\leq t\leq m-1$, similar discussion as Case II.4 shows that $[3i\cdot 2^{t}]_{3n}>3i$.
		
		{\bf Case II.6.} When $t=m$, we have
		\begin{align*}
			[3i\cdot 2^{m}]_{3n}&=3n-3i=(1,1,\cdots,1,3)-(0,0,\overset{m-9}{\overbrace{1,0,\cdots,1,0}},1,0,0,1,0,1,1)\\
			&=(1,1,\overset{m-9}{\overbrace{0,1,\cdots,0,1}},0,1,1,0,1,0,2)>3i.
		\end{align*}
		
		{\bf Case II.7.} When $m+1\leq t\leq m+2$, similar discussion as Case II.6 shows that $[3i\cdot 2^{t}]_{3n}>3i$.
		
		{\bf Case II.8.} Suppose $m+3\leq t\leq 2m-7$. Let $\ell=t-m$, then $3\leq \ell\leq m-7$. If $t$ is even, then $\ell$ is odd and
		\begin{align*}
			[3i\cdot 2^{t}]_{3n}&=3n-[3i\cdot 2^{\ell}]_{3n}\\
			&=(1,1,\cdots,1,3)-(0,\overset{m-8-\ell}{\overbrace{1,0,\cdots,1,0}},1,0,0,1,0,1,1,0,0,-1,\overset{\ell-3}{\overbrace{0,-1,\cdots,0,-1}})\\
			&\geq (1,1,\cdots,1,3)-(0,\overset{m-11}{\overbrace{1,0,\cdots,1,0}},1,0,0,1,0,1,1,0,0,-1)\\
			&=(1,\overset{m-11}{\overbrace{0,1,\cdots,0,1}},0,1,1,0,1,0,0,1,1,4)>3i.
		\end{align*}
		If $t$ is odd, then $\ell$ is even and
		\begin{align*}
			[3i\cdot 2^{t}]_{3n}&=3n-[3i\cdot 2^{l}]_{3n}\\
			&=(1,1,\cdots,1,3)-(\overset{m-7-\ell}{\overbrace{1,0,\cdots,1,0}},1,0,0,1,0,1,1,0,0,\overset{\ell-2}{\overbrace{-1,0,\cdots,-1,0}})\\
			&\geq (1,1,\cdots,1,3)-(\overset{m-11}{\overbrace{1,0,\cdots,1,0}},1,0,0,1,0,1,1,0,0,-1,0)\\
			&=(\overset{m-11}{\overbrace{0,1,\cdots,0,1}},0,1,1,0,1,0,0,1,1,2,3)>3i.
		\end{align*}
		
		{\bf Case II.9.} When $t=2m-6$, we have
		\begin{align*}
			[3i\cdot 2^{2m-6}]_{3n}&=3n-[3i\cdot 2^{m-6}]_{3n}\\
			&=(1,1,\cdots,1,3)-(0,0,1,0,1,1,0,0,\overset{m-9}{\overbrace{-1,0,\cdots,-1,0}},-1)\\
			&=(1,1,0,1,0,0,1,1,\overset{m-9}{\overbrace{2,1,\cdots,2,1}},4)>3i.
		\end{align*}
		
		{\bf Case II.10.} When $2m-5\leq t\leq 2m-1$, similar discussion as Case II.9 shows that $[3i\cdot 2^{t}]_{3n}>3i$.
		
		The proof is then completed.
	\end{proof}
	
		\begin{Theorem}\label{t4.1}
		Let $q=2$ and $n=\frac{2^{m}+1}{3}$, where $m$ is an odd integer. Suppose $m\equiv 0~({\rm mod}~3)$, and let $\delta_{1}$ and $\delta_{2}$ be as given in Proposition \ref{p4.1}.
		\begin{itemize}[align=left,leftmargin=*]
			\item[{\rm (1)}] For $\delta_{2}+1\leq \delta\leq \delta_{1}$, the codes $\mathcal{C}_{(q,n,\delta,1)}$ and $\mathcal{C}_{(q,n,\delta+1,0)}$ have parameters $[\frac{2^{m}+1}{3},3,\frac{2^{m}+1}{9}]$ and $[\frac{2^{m}+1}{3},2,\frac{2^{m+1}+2}{9}]$, respectively.
			
			\item[{\rm (2)}] The codes $\mathcal{C}_{(q,n,\delta_{2},1)}$ and and $\mathcal{C}_{(q,n,\delta_{2}+1,0)}$ have parameters $[\frac{2^{m}+1}{3}, 2m+3, \geq \frac{2^{m-1}-31}{9}]$ and
			$[\frac{2^{m}+1}{3}, 2m+2, \geq \frac{2^{m}-62}{9}]$, respectively.
		\end{itemize}
	\end{Theorem}
	\begin{proof}
		The proof follows directly from Lemmas \ref{l2.3}-\ref{l2.5} and Proposition \ref{p4.1}.
	\end{proof}
	
	\begin{Example}{\rm
		Take $m=9$ in Theorem \ref{t4.1}, then $n=171$. According to Proposition \ref{p4.1}, $\delta_{1}=57$ and $\delta_{2}=25$. 
		\begin{itemize}[align=left,leftmargin=*]
		\item[(i)] By Theorem \ref{t4.1}, the codes $\mathcal{C}_{(2,n,\delta_{1},1)}$ and $\mathcal{C}_{(2,n,\delta_{1}+1,0)}$ have parameters $[171,3,57]$ and $[171,2,114]$, respectively. The code $\mathcal{C}_{(2,n,\delta_{1}+1,0)}$ meets the Griesmer bound and has the same parameters as the optimal binary linear code known in the Database.
		
		\item[(ii)] By Theorem \ref{t4.1}, the codes $\mathcal{C}_{(2,n,\delta_{2},1)}$ and $\mathcal{C}_{(2,n,\delta_{2}+1,0)}$ have parameters $[171,21,\geq 25]$ and $[171,20,\geq 50]$, respectively. 
		\end{itemize}}
	\end{Example}
	
	The following proposition uses similar proof methods as in the proof of Proposition \ref{p4.1}.
		\begin{Proposition}\label{p4.2}
		Let $q=2$ and $n=\frac{2^{m}+1}{3}$, where $m$ is an odd integer. Suppose $m\equiv 1~({\rm mod}~3)$, then:
		
		\noindent {\rm (1)} if $m=7$, then $\delta_{1}=7$ and $\delta_{2}=3$; moreover, $|C_{\delta_{1}}|=|C_{\delta_{2}}|=14$.
		
		\noindent {\rm (2)} if $m>7$, then $\delta_{1}=\frac{2^{m-1}-1}{9}$ and $\delta_{2}=\frac{2^{m-1}-127}{9}$; moreover, $|C_{\delta_{1}}|=|C_{\delta_{2}}|=2m$.
	\end{Proposition}
	
	    	\begin{Theorem}\label{t4.2}
		Let $q=2$ and $n=\frac{2^{m}+1}{3}$, where $m$ is an odd integer. Suppose $m\equiv 1~({\rm mod}~3)$, and let $\delta_{1}$ and $\delta_{2}$ be as given in Proposition \ref{p4.2}.
		\begin{itemize}[align=left,leftmargin=*]
			\item[{\rm (1)}] For $\delta_{2}+1\leq \delta\leq \delta_{1}$, the codes $\mathcal{C}_{(q,n,\delta,1)}$ and $\mathcal{C}_{(q,n,\delta+1,0)}$ have parameters $[\frac{2^{m}+1}{3},2m+1,\geq \frac{2^{m-1}-1}{9}]$ and $[\frac{2^{m}+1}{3},2m,\geq \frac{2^{m}-2}{9}]$, respectively. 
			
			\item[{\rm (2)}] When $m=7$, the codes $\mathcal{C}_{(q,n,\delta_{2},1)}$ and $\mathcal{C}_{(q,n,\delta_{2}+1,0)}$ have parameters $[43,29,\geq 3]$ and $[43,28,\geq 6]$, respectively. When $m>7$, the codes $\mathcal{C}_{(q,n,\delta_{2},1)}$ and $\mathcal{C}_{(q,n,\delta_{2}+1,0)}$ have parameters $[\frac{2^{m}+1}{3}, 4m+1, \geq \frac{2^{m-1}-127}{9}]$ and $[\frac{2^{m}+1}{3}, 4m, \geq \frac{2^{m}-254}{9}]$, respectively.
		\end{itemize}
	\end{Theorem}
	\begin{proof}
		The conclusions are straightforward from Lemmas \ref{l2.3}, \ref{l2.4} and Proposition \ref{p4.1}.
	\end{proof}
	
	\begin{Example}{\rm
		Take $m=7$ in Theorem \ref{t4.2}, then $n=43$. According to Proposiyion \ref{p4.2}, $\delta_{1}=7$ and $\delta_{2}=3$.
		\begin{itemize}[align=left,leftmargin=*]
			\item[(i)] By Theorem \ref{t4.2}, the codes $\mathcal{C}_{(2,n,\delta_{1},1)}$ and $\mathcal{C}_{(2,n,\delta_{1}+1,0)}$ have parameters $[43,15,\geq 7]$ and $[43,14,\geq 14]$, respectively. By use of Magma, $\mathcal{C}_{(2,n,\delta_{1},1)}$ and $\mathcal{C}_{(2,n,\delta_{1}+1,0)}$ have parameters $[43,15,13]$ and $[43,14,14]$, respectively. The code $\mathcal{C}_{(2,n,\delta_{1},1)}$ has the same parameters as the best binary linear code in the Database, and the code $\mathcal{C}_{(2,n,\delta_{1}+1,0)}$ has the same parameters as the optimal binary linear code in the Database.
			
			\item[(ii)] By Theorem \ref{t4.2}, the codes $\mathcal{C}_{(2,n,\delta_{2},1)}$ and $\mathcal{C}_{(2,n,\delta_{2}+1,0)}$ have parameters $[43,29,\geq 3]$ and $[43,28,\geq 6]$, respectively. By use of Magma, $\mathcal{C}_{(2,n,\delta_{2},1)}$ and $\mathcal{C}_{(2,n,\delta_{2}+1,0)}$ have parameters $[43,29,6]$ and $[43,28,6]$, respectively. The code $\mathcal{C}_{(2,n,\delta_{2},1)}$ has the same parameters as the optimal binary linear code in the Database, which is not known to be cyclic. The code $\mathcal{C}_{(2,n,\delta_{2}+1,0)}$ has the same parameters as the best binary linear code in the Database, which is not known to be cyclic.
		\end{itemize}}
	\end{Example}

	The proof of the proposition below follows the same path as that of Proposition \ref{p4.2}.
	\begin{Proposition}\label{p4.3}
	Let $q=2$ and $n=\frac{2^{m}+1}{3}$, where $m$ is an odd integer. Suppose $m\equiv 2~({\rm mod}~3)$, then:
	
	\noindent {\rm (1)} if $m=5$, then $\delta_{1}=1$ with $|C_{\delta_{1}}|=10$.
	
	\noindent {\rm (2)} if $m>5$, then $\delta_{1}=\frac{2^{m-1}-7}{9}$ and $\delta_{2}=\frac{2^{m-1}-25}{9}$; moreover, $|C_{\delta_{1}}|=|C_{\delta_{2}}|=2m$.
	\end{Proposition}
    
    The following results can be deduced from Lemmas \ref{l2.3}, \ref{l2.4} and Proposition \ref{p4.3}.
    	\begin{Theorem}\label{t4.3}
    	Let $q=2$ and $n=\frac{2^{m}+1}{3}$, where $m$ is an odd integer. Suppose $m\equiv 2~({\rm mod}~3)$, and let $\delta_{1}$ and $\delta_{2}$ be as given in Proposition \ref{p4.3}.
    	\begin{itemize}[align=left,leftmargin=*]
    	\item[{\rm (1)}] For $\delta_{2}+1\leq \delta\leq \delta_{1}$, the codes $\mathcal{C}_{(q,n,\delta,1)}$ and $\mathcal{C}_{(q,n,\delta+1,0)}$ have parameters $[\frac{2^{m}+1}{3},2m+1,\geq \frac{2^{m-1}-7}{9}]$ and $[\frac{2^{m}+1}{3},2m,\geq \frac{2^{m}-14}{9}]$, respectively. 
    	
    	\item[{\rm (2)}] The codes $\mathcal{C}_{(q,n,\delta_{2},1)}$ and $\mathcal{C}_{(q,n,\delta_{2}+1,0)}$ have parameters $[\frac{2^{m}+1}{3}, 4m+1, \geq \frac{2^{m-1}-25}{9}]$ and $[\frac{2^{m}+1}{3}, 4m, \geq \frac{2^{m}-50}{9}]$, respectively.
    \end{itemize}
    \end{Theorem}
    
    \begin{Example}{\rm
    	Take $m=11$ in Theorem \ref{t4.3}, then $n=683$. According to Proposition \ref{p4.3}, $\delta_{1}=113$ and $\delta_{2}=111$. By Theorem \ref{t4.3}, the codes $\mathcal{C}_{(2,n,\delta_{1},1)}$ and $\mathcal{C}_{(2,n,\delta_{1}+1,0)}$ have parameters $[683,23,\geq 113]$ and $[683, 22,\geq 226]$, respectively; the codes $\mathcal{C}_{(2,n,\delta_{2},1)}$ and $\mathcal{C}_{(2,n,\delta_{2}+1,0)}$ have parameters $[683, 45, \geq 111]$ and $[683, 44, \geq 222]$, respectively.}
    \end{Example}
    
    \section{The dual codes of ternary antiprimitive BCH codes}
    In this section, we always assume that $q=3$ and $n=3^m+1$, where $m>1$ is an integer. Our task in this section is to develop a lower bound on the minimum distance of the dual code $\mathcal{C}_{(q,n,\delta,1)}^{\bot}$. To this end, we need the lemma below.
    
    \begin{Lemma}{\rm \cite{zswh}}
    	Suppose $q=3$ and $n=3^m+1$, where $m>1$. The largest $q$-coset leader modulo $n$ is $\delta_{1}=\frac{3^m+1}{2}$.
    \end{Lemma}
    
    Recall that the defining set of $\mathcal{C}_{(q,n,\delta,1)}$ is $T=C_{1}\cup C_{2}\cup\cdots \cup C_{\delta-1}$, where $2\leq \delta \leq n$. Denote by $T^{\bot}$ the defining set of $\mathcal{C}_{(q,n,\delta,1)}^{\bot}$. It is clear that $T^{\bot}=\mathbb{Z}_{n}\setminus (-T)=\mathbb{Z}_{n}\setminus T$.

    For $2\leq \delta \leq \delta_1$, let $1\leq I_{1}(\delta)<\delta_1$ be the integer such that $I_{1}(\delta) \in T$ and $\{I_{1}(\delta)+1, I_{1}(\delta)+2,\cdots, \delta_1\}\subseteq \mathbb{Z}_{n}\setminus T$, and let $\delta_1+1\leq I_{2}(\delta)\leq n-1$ be the integer such that $I_{2}(\delta) \in T$ and $\{\delta_1,\delta_1+1,\cdots, I_{2}(\delta)-1\}\subseteq \mathbb{Z}_{n}\setminus T$. By the BCH bound for cyclic codes, we have $d(\mathcal{C}_{(q,n,\delta,1)}^{\bot})\geq I_{2}(\delta)-I_{1}(\delta)$. Next we determine the values of $I_{1}(\delta)$ and $I_{2}(\delta)$.
    
    \begin{Proposition}\label{p5.1}
    	Suppose $q=3$ and $n=3^m+1$, where $m>1$. Let $1\leq \ell\leq m-1$. For $2\leq \delta\leq \frac{3^m+1}{2}$, we have
    	\begin{align*}
    		I_{1}(\delta)=\begin{cases}
    			\frac{3^m-3^{m-\ell}}{2},& {\rm if}~\delta=\frac{3^\ell+1}{2},\\
    			\frac{3^m-3^{m-\ell}}{2}+1,& {\rm if}~\frac{3^\ell+3}{2}\leq \delta\leq \frac{3^{\ell+1}-1}{2},\\
    			\frac{3^m-1}{2},& {\rm if}~\delta=\frac{3^m+1}{2},
    			\end{cases}
    	\end{align*}
    	and $I_{2}(\delta)=n-I_{1}(\delta)$.
    \end{Proposition}
    \begin{proof}
    	When $\delta=\frac{3^\ell+1}{2}$, it is easy to see that $\frac{3^m-3^{m-\ell}}{2}=\frac{(3^{\ell}-1)3^{m-\ell}}{2}\in C_{\frac{3^{\ell}-1}{2}}\subseteq T$. Now we are going to show that $\{\frac{3^m-3^{m-\ell}}{2}+1,\frac{3^m-3^{m-\ell}}{2}+2,\cdots,\delta_{1}\}\subseteq \mathbb{Z}_{n}\setminus T$. For every integer $i$ with $\frac{3^m-3^{m-\ell}}{2}+1\leq i\leq \delta_{1}$, we need to prove that the coset leader of the $q$-coset containing $i$ is larger than or equal to $\frac{3^\ell+1}{2}$. To this end, we next prove that $[i\cdot 3^{t}]_{n}\geq \frac{3^\ell+1}{2}$ for every $0\leq t\leq 2m-1$.
    	
    	For $i\in \{\frac{3^m-3^{m-\ell}}{2}+1,\frac{3^m-3^{m-\ell}}{2}+2,\cdots,\delta_{1}\}$, write $i=\frac{3^m-3^{m-\ell}}{2}+u$, where $1\leq u\leq \delta_{1}-\frac{3^m-3^{m-\ell}}{2}=\frac{3^{m-\ell}+1}{2}$. Note that
    	$$\frac{3^m-3^{m-\ell}}{2}=(\overset{\ell}{\overbrace{1,1,\cdots,1}},\overset{m-\ell}{\overbrace{0,0,\cdots,0}}),$$
    	$$\frac{3^{m-\ell}+1}{2}=(\overset{\ell}{\overbrace{0,0,\cdots,0}},\overset{m-\ell-1}{\overbrace{1,1,\cdots,1}},2),~{\rm and}~\frac{3^\ell+1}{2}=(\overset{m-\ell}{\overbrace{0,0,\cdots,0}},\overset{\ell-1}{\overbrace{1,1,\cdots,1}},2).$$
    	Write $u=(\overset{\ell}{\overbrace{0,0,\cdots,0}},i_{m-\ell-1},i_{m-\ell-2},\cdots,i_{0})$, where $i_{j}\in\{0,1,2\}$ for $0\leq j\leq m-\ell-1$. Then $i=\frac{3^m-3^{m-\ell}}{2}+u=(\overset{\ell}{\overbrace{1,1,\cdots,1}},i_{m-\ell-1},i_{m-\ell-2},\cdots,i_{0})$. We organize the proof into five cases.
    	
    	{\bf Case 1.} When $t=0$, $[i\cdot 3^0]_{n}=i>\frac{3^\ell+1}{2}$.
    	
    	{\bf Case 2.} When $1\leq t\leq \ell-1$, we have
    	\begin{align*}
    		[i\cdot 3^t]_{n}&=(\overset{\ell-t}{\overbrace{1,1,\cdots,1}},i_{m-\ell-1},i_{m-\ell-2},\cdots,i_{0},\overset{t}{\overbrace{-1,-1,\cdots,-1}})\\
    		&\geq (1,\overset{m-\ell-1}{\overbrace{0,0,\cdots,0}},1,\overset{\ell-1}{\overbrace{-1,-1,\cdots,-1}})=(1,\overset{m-\ell}{\overbrace{0,0,\cdots,0}},\overset{\ell-2}{\overbrace{1,1,\cdots,1}},2)>\frac{3^\ell+1}{2}.
    	\end{align*}
    	
    	{\bf Case 3.} When $\ell\leq t\leq m-1$, we have
    	\begin{align*}
    		[i\cdot 3^t]_{n}&=(i_{m-t-1},\cdots,i_{1},i_{0},\overset{\ell}{\overbrace{-1,-1,\cdots,-1}},-i_{m-\ell-1},\cdots,-i_{m-t})\\
    		&\geq (\overset{m-t-1}{\overbrace{0,0,\cdots,0}},1,\overset{t}{\overbrace{-1,-1,\cdots,-1}})\\
    		&=(\overset{m-t}{\overbrace{0,0,\cdots,0}},\overset{t-1}{\overbrace{1,1,\cdots,1}},2)\geq (\overset{m-\ell}{\overbrace{0,0,\cdots,0}},\overset{\ell-1}{\overbrace{1,1,\cdots,1}},2)>\frac{3^\ell+1}{2}.
    	\end{align*}
    	
    	{\bf Case 4.} When $m\leq t\leq m+\ell$, let $l=t-m$, then $0\leq l\leq \ell$ and
    		\begin{align*}
    		[i\cdot 3^t]_{n}&=n-[i\cdot 3^l]_{n}\\
    		&=(2,2,\cdots,2,4)-(\overset{\ell-l}{\overbrace{1,1,\cdots,1}},i_{m-\ell-1},i_{m-\ell-2},\cdots,i_{0},\overset{l}{\overbrace{-1,-1,\cdots,-1}})\\
    		&\geq (2,2,\cdots,2,4)-(1,1,\cdots,1,2)=(1,1,\cdots,1,2)>\frac{3^\ell+1}{2}.
    		\end{align*}
    		
    		{\bf Case 5.} When $m+\ell+1\leq t\leq 2m-1$, let $l=t-m$, then $\ell+1\leq l\leq m-1$ and
    		\begin{align*}
    			[i\cdot 3^t]_{n}&=n-[i\cdot 3^l]_{n}\\
    			&=(2,2,\cdots,2,4)-(i_{m-l-1},\cdots,i_{1},i_{0},\overset{\ell}{\overbrace{-1,-1,\cdots,-1}},-i_{m-\ell-1},\cdots,-i_{m-l})\\
    			&\geq (2,2,\cdots,2,4)-(\overset{m-l}{\overbrace{2,2,\cdots,2}},\overset{\ell}{\overbrace{-1,-1,\cdots,-1}},\overset{l-\ell}{\overbrace{0,0,\cdots,0}})\\
    			&=(\overset{m-l}{\overbrace{0,0,\cdots,0}},\overset{\ell}{\overbrace{3,3,\cdots,3}},\overset{l-\ell-1}{\overbrace{2,2,\cdots,2}},4)\\
    			&=(\overset{m-l-1}{\overbrace{0,0,\cdots,0}},\overset{\ell+1}{\overbrace{1,1,\cdots,1}},\overset{l-\ell-1}{\overbrace{0,0,\cdots,0}},1)\geq (\overset{m-\ell-2}{\overbrace{0,0,\cdots,0}},\overset{\ell+2}{\overbrace{1,1,\cdots,1}})>\frac{3^\ell+1}{2}.
    		\end{align*}
    		Hence $I_{1}(\delta)=\frac{3^m-3^{m-\ell}}{2}$ when $\delta=\frac{3^\ell+1}{2}$.
    		
    		Suppose $\frac{3^\ell+3}{2}\leq \delta\leq \frac{3^{\ell+1}-1}{2}$. It is easily seen that $\frac{3^m-3^{m-\ell}}{2}+1\in C_{n-\frac{3^m-3^{m-\ell}}{2}-1}$. As $n-\frac{3^m-3^{m-\ell}}{2}-1=\frac{3^m+3^{m-\ell}}{2}=\frac{(3^{\ell}+1)3^{m-\ell}}{2}$, $\frac{3^m-3^{m-\ell}}{2}+1\in C_{\frac{3^{\ell}+1}{2}}\subseteq T$. To verify  $\{\frac{3^m-3^{m-\ell}}{2}+2,\frac{3^m-3^{m-\ell}}{2}+3,\cdots,\delta_{1}\}\subseteq \mathbb{Z}_{n}\setminus T$, we need to show that for every integer $i$ with $\frac{3^m-3^{m-\ell}}{2}+2\leq i\leq \delta_{1}$, the coset leader of the $q$-coset containing $i$ is larger than or equal to $\frac{3^{\ell+1}-1}{2}$, whose proof is similar as above and so omitted here. When $\delta=\frac{3^m+1}{2}$, $\frac{3^m-1}{2}\in T$ and $\frac{3^m+1}{2}\in \mathbb{Z}_{n}\setminus T$, thus $I_{1}(\delta)=\frac{3^m-1}{2}$.
    		
    		As for the value of $I_{2}(\delta)$, we first have $n-I_{1}(\delta)\in C_{I_{1}(\delta)}\subseteq T$. Note that $\delta_{1}=\frac{n}{2}$. For every integer $i\in \{\delta_{1},\delta_{1}+1,\cdots,n-I_{1}(\delta)-1\}$, $n-i\in \{I_{1}(\delta)+1,I_{1}(\delta)+2,\cdots,\delta_{1}\}\subseteq \mathbb{Z}_{n}\setminus T$, and hence $i\in C_{n-i}\subseteq \mathbb{Z}_{n}\setminus T$. It follows that $I_{2}(\delta)=n-I_{1}(\delta)$.
    \end{proof}

    \begin{Theorem}\label{t5.1}
    	Suppose $q=3$ and $n=3^m+1$, where $m>1$. Let $1\leq \ell\leq m-1$. For $2\leq \delta\leq n$, we have
    	\begin{align*}
    		d(\mathcal{C}_{(q,n,\delta,1)}^{\bot})\geq \begin{cases}
    			3^{m-\ell}+1, & {\rm if}~\delta=\frac{3^\ell+1}{2},\\
    			3^{m-\ell}-1, & {\rm if}~\frac{3^\ell+3}{2}\leq \delta\leq \frac{3^{\ell+1}-1}{2},\\
    			2, & {\rm if}~\frac{3^m+1}{2}\leq \delta \leq n.
    			\end{cases}
    	\end{align*}
    \end{Theorem}
    \begin{proof}
    	The desired result follows from Lemma \ref{l2.3} and Proposition \ref{p5.1}.
    \end{proof}
    
    The following example shows that the lower bound in Theorem \ref{t5.1} is good.
    \begin{Example}{\rm
    	Take $m=3$ in Theorem \ref{t5.1}, then $n=28$. For $2\leq \delta \leq n$, the lower bound on the minimum distance of $\mathcal{C}^{\perp}_{(3,n,\delta,1)}$ given by Theorem \ref{t5.1} and the minimum distance of $\mathcal{C}^{\perp}_{(3,n,\delta,1)}$ obtained by use of Magma are listed in Table 1.
    	\begin{longtable}{|c|c|c|}
    		\hline
    		$\delta$ & $d(\mathcal{C}_{(3,n,\delta,1)}^{\bot})\geq$ & $d(\mathcal{C}_{(3,n,\delta,1)}^{\bot})$  \\
    		\hline
    		$2$ & $10$ & $12$ \\
    		\hline
    		$3\sim4$ & $8$ & $8$ \\
    		\hline
    		$5$ & $4$ & $4$ \\
    		\hline
    		$6\sim28$ & $2$ & $2$ \\
    		\hline
    		\caption{A lower bound on the minimum distance of $\mathcal{C}^{\perp}_{(3,n,\delta,1)}$}
    	\end{longtable}}
    \end{Example}
    
    \section{Concluding remarks and future works}
    In this paper, we extend the previous works and develop some new results on $q$-ary BCH codes with lengths $n=\frac{q^m+1}{q+1}$ and $n=q^m+1$. For $n=\frac{q^m+1}{q+1}$, we settle the dimensions of narrow-sense  BCH codes with designed distance $\delta=\ell q^{\frac{m-1}{2}}+1$, where $q>2$ and $2\leq \ell \leq q-1$. Moreover, we determine the largest coset leader for $m=3$ and also determine the first two largest coset leaders for $q=2$. Based on these results, we investigate the parameters of BCH codes with large designed distance and we find out some optimal BCH codes. For tenary narrow-sense BCH codes of length $n=3^m+1$, we derive a tight lower bound on the minimum distance of their dual codes. Some techniques presented in this paper may be helpful to the further research on $q$-ary BCH codes of length $n=\frac{q^m+1}{N}$. 
    
    A possible direction for future work is to find the first few largest coset leaders modulo $n=\frac{q^m+1}{q+1}$ for general $q$ and $m$.
    
    	\noindent\textbf{Acknowledgement.}
    	
    	This work was supported by National Natural Science Foundation of China
    	under Grant Nos.12271199 and 12171191 and The Fundamental Research Funds for the Central Universities
    	30106220482.

	 \section*{Appendix}
	
	\noindent{\bf Proof of Proposition \ref{p3.4}.}
	We are now going to find out all integers $a$ with $q^{h}\leq a\leq (q-1)q^{h}$ and $a\not\equiv 0~({\rm mod}~q)$ that are not coset leaders. In other words, for every $1\leq t\leq 2m-1=4h+1$, we need to find out all integers $a$ with $q^{h}\leq a\leq (q-1)q^{h}$ and $a\not\equiv 0~({\rm mod}~q)$ satisfying $[aq^{t}]_{n}<a$. By assumption we have $1\leq a_{h}\leq q-2$ and $1\leq a_{0}\leq q-1$. We organize the proof into ten cases.
	
	{\bf Case 1.} When $1\leq t\leq h-1$, it is clear that $a<aq^{t}\leq (q-1)q^{h}\cdot q^{h-1}=(q-1)q^{2h-1}<n.$
	
	{\bf Case 2.} When $t=h$, we have
	\begin{align*}
		aq^{h}=\big(\sum_{i=0}^{h}a_{i}q^{i}\big)q^{h}=\sum_{i=h}^{2h}a_{i-h}q^{i}\equiv a_{h}\sum_{i=0}^{2h-1}(-1)^{i+1}q^{i}+\sum_{i=h}^{2h-1}a_{i-h}q^{i}~({\rm mod}~n).
	\end{align*}
	Denote $M_{h}=a_{h}\sum\limits_{i=0}^{2h-1}(-1)^{i+1}q^{i}+\sum\limits_{i=h}^{2h-1}a_{i-h}q^{i}$. It is obvious that $M_{h}>0$. Moreover, 
	$$M_{h}\leq (q-2)\sum_{i=0}^{2h-1}(-1)^{i+1}q^{i}+\sum_{i=h}^{2h-1}(q-1)q^{i}=\frac{2q^{2h+1}-q^{2h}-q^{h+1}-q^{h}-q+2}{q+1}<2n.$$
	If $0<M_{h}<n$, we have
	$$[aq^{h}]_{n}=M_{h}\geq \sum_{i=0}^{2h-1}(-1)^{i+1}q^{i}+q^{h}=\frac{q^{2h}+q^{h+1}+q^{h}-1}{q+1}>(q-1)q^{h}>a.$$
	Suppose $n<M_{h}<2n$. Then 
	\begin{align*}
		[aq^{h}]_{n}&=M_{h}-n=a_{h}\sum_{i=0}^{2h-1}(-1)^{i+1}q^{i}+\sum_{i=h}^{2h-1}a_{i-h}q^{i}+\sum_{i=0}^{2h}(-1)^{i+1}q^{i}\\
		&=(a_{h}+a_{h-1}-q+1)q^{2h-1}+(a_{h}+1)\sum_{i=0}^{2h-2}(-1)^{i+1}q^{i}+\sum_{i=h}^{2h-2}a_{i-h}q^{i}.
	\end{align*}
	If $a_{h}+a_{h-1}<q-1$, then
	$$[aq^{h}]_{n}\leq -q^{2h-1}+2\sum_{i=0}^{2h-2}(-1)^{i+1}q^{i}+\sum_{i=h}^{2h-2}(q-1)q^{i}=-\frac{2q^{2h-1}+q^{h+1}+q^{h}+2}{q+1}<0,$$
	which is impossible. If $a_{h-1}+a_{h}>q-1$, then
	$$[aq^{h}]_{n}\geq q^{2h-1}+(q-1)\sum_{i=0}^{2h-2}(-1)^{i+1}q^{i}+q^{h}=\frac{2q^{2h-1}+q^{h+1}+q^{h}-q+1}{q+1}>(q-1)q^{h}>a.$$
	So $a_{h}+a_{h-1}=q-1$, and therefore
	\begin{align*}
		[aq^{h}]_{n}&=(a_{h}+1)\sum_{i=0}^{2h-2}(-1)^{i+1}q^{i}+\sum_{i=h}^{2h-2}a_{i-h}q^{i}\\
		&=\sum_{i=h}^{2h-2}\big[(-1)^{i+1}(a_{h}+1)+a_{i-h}\big]q^{i}+(a_{h}+1)\sum_{i=0}^{h-1}(-1)^{i+1}q^{i}\\
		&=\sum_{i=(h+1)/2}^{h-1}\Delta_{i,h}q^{2i-1}-(a_{h}+1)\frac{q^{h}+1}{q+1},
	\end{align*}
	where $\Delta_{i,h}=(-a_{h}-1+a_{2i-h})q+a_{h}+1+a_{2i-1-h}$. For $\frac{h+1}{2}\leq i\leq h-1$, one can check that $$-(q-1)^{2}\leq \Delta_{i,h}=-(a_{h}+1)(q-1)+a_{2i-h}q+a_{2i-1-h}\leq (q-1)^{2}.$$
	
	Suppose that $\Delta_{i,h}\neq 0$ for some $\frac{h+3}{2}\leq i\leq h-1$. Let $\frac{h+3}{2}\leq l\leq h-1$ be the largest integer such that $\Delta_{l,h}\neq 0$. If $\Delta_{l,h}<0$, then
    $$[aq^{h}]_{n}\leq -q^{2l-1}+\sum_{i=(h+1)/2}^{l-1}(q-1)^{2}q^{2i-1}-\frac{2(q^{h}+1)}{q+1}=-\frac{2q^{2l-1}+q^{h+1}+q^{h}+2}{q+1}<0,$$
	which is impossible. If $\Delta_{l,h}>0$, then
	\begin{align*}
		[aq^{h}]_{n}&\geq q^{2l-1}-\sum_{i=(h+1)/2}^{l-1}(q-1)^{2}q^{2i-1}-\frac{(q-1)(q^{h}+1)}{q+1}\\
		&=\frac{2q^{2l-1}-q+1}{q+1}\geq \frac{2q^{h+2}-q+1}{q+1}>(q-1)q^{h}> a.
	\end{align*}
	Hence for $\frac{h+3}{2}\leq i\leq h-1$, $\Delta_{i,h}=0$, or equivalently, $(a_{h}+1-a_{2i-h})q=a_{h}+1+a_{2i-1-h}$. As $0<a_{h}+a_{2i-1-h}+1\leq 2q-2$, $a_{h}+1-a_{2i-h}=1$ and $a_{h}+a_{2i-1-h}+1=q$ implying $a_{2i-h}=a_{h}$ and $a_{2i-1-h}=q-1-a_{h}$. Consequently, 
	$$[aq^{h}]_{n}=\Delta_{\frac{h+1}{2},h}q^{h}-(a_{h}+1)\frac{q^{h}+1}{q+1},$$
	where $\Delta_{\frac{h+1}{2},h}=(-a_{h}-1+a_{1})q+a_{h}+1+a_{0}$, and 
	\begin{align*}
		a&=a_{h}(q^{h}+q^{h-2}+\cdots+q^{3})+(q-1-a_{h})(q^{h-1}+q^{h-3}+\cdots+q^{2})+a_{1}q+a_{0}\\
		&=(a_{h}+1)\frac{q^{h+1}-q^{2}}{q+1}+a_{1}q+a_{0}.
	\end{align*}
	If $\Delta_{\frac{h+1}{2},h}\leq 0$, then $[aq^{h}]_{n}<0$, which is impossible; so $\Delta_{\frac{h+1}{2},h}>0$, or equivalently, $(a_{h}+1-a_{1})q<a_{h}+1+a_{0}$. Since $3\leq a_{h}+1+a_{0}\leq 2q-2$, $\Delta_{\frac{h+1}{2},h}>0$ if and only if one of the following holds:
	
	\noindent (i) $a_{h}+1-a_{1}=1$ and $q<a_{h}+1+a_{0}$,
	
	\noindent (ii) $a_{h}+1-a_{1}\leq 0$.
	
	Assume (i). Then $a_{1}=a_{h}$ and $a_{h}+a_{0}\geq q$. As a consequence,
	$$[aq^{h}]_{n}=(-q+a_{h}+1+a_{0})q^{h}-(a_{h}+1)\frac{q^{h}+1}{q+1}=\frac{(a_{h}+a_{0}-q)q^{h+1}+a_{0}q^{h}-a_{h}-1}{q+1},$$
	and
	$$a=(a_{h}+1)\frac{q^{h+1}-q^{2}}{q+1}+a_{h}q+a_{0}=\frac{(a_{h}+1)q^{h+1}-q^{2}+(a_{h}+a_{0})q+a_{0}}{q+1}.$$
	It follows that
	$$[aq^{h}]_{n}<a~\Leftrightarrow~-(q-a_{0})q^{h+1}-(q-a_{0})q^{h}-(a_{h}+a_{0}-q)q-a_{h}-a_{0}-1<0.$$
	The right-hand side inequality above is trivially true.
	
	Assume (ii). If $a_{h}+1-a_{1}\leq -1$, then $\Delta_{\frac{h+1}{2},h}\geq q+3$, which leads to
	$$[aq^{h}]_{n}\geq (q+3)q^{h}-(a_{h}+1)\frac{q^{h}+1}{q+1}>(q-1)q^{h}>a.$$
	So $a_{h}+1-a_{1}=0$. Hence
	$$[aq^{h}]_{n}=(a_{h}+1+a_{0})q^{h}-(a_{h}+1)\frac{q^{h}+1}{q+1}=\frac{(a_{h}+1+a_{0})q^{h+1}+a_{0}q^{h}-a_{h}-1}{q+1},$$
	and
    $$a=(a_{h}+1)\frac{q^{h+1}-q^{2}}{q+1}+(a_{h}+1)q+a_{0}=\frac{(a_{h}+1)q^{h+1}+(a_{h}+1+a_{0})q+a_{0}}{q+1}.$$
    However, it is easy to see that $[aq^{h}]_{n}>a$.
	
	In conclusion, $[aq^{h}]_{n}<a$ if and only if $a$ satisfies the following condition:
	
	\noindent {\bf c1)} $a=(a_{h}+1)\frac{q^{h+1}+q}{q+1}-q+a_{0}$ with $1\leq a_{h}\leq q-2$ and $q-a_{h}\leq a_{0}\leq q-1$.
	
	{\bf Case 3.} When $h+1\leq t\leq 2h-2$, we have
	\begin{align*}
		aq^{t}&=\big(\sum_{i=0}^{h}a_{i}q^{i}\big)q^{t}=\sum_{i=t}^{h+t}a_{i-t}q^{i}=\sum_{i=2h+1}^{h+t}a_{i-t}q^{i}+a_{2h-t}q^{2h}+\sum_{i=t}^{2h-1}a_{i-t}q^{i}\\
		&\equiv -\sum_{i=0}^{t-h-1}a_{i+2h+1-t}q^{i}+a_{2h-t}\sum_{i=0}^{2h-1}(-1)^{i+1}q^{i}+\sum_{i=t}^{2h-1}a_{i-t}q^{i}~({\rm mod}~n).
	\end{align*}
	Denote $M_{t}=-\sum\limits_{i=0}^{t-h-1}a_{i+2h+1-t}q^{i}+a_{2h-t}\sum\limits_{i=0}^{2h-1}(-1)^{i+1}q^{i}+\sum\limits_{i=t}^{2h-1}a_{i-t}q^{i}$. One can check that $0<M_{t}<2n$ and if $0<M_{t}<n$, then $[aq^{t}]_{n}=M_{t}>(q-1)q^{h}>a$.
	Suppose $n<M_{t}<2n$. Then
	\begin{align*}
		[aq^{t}]_{n}&=M_{t}-n= -\sum_{i=0}^{t-h-1}a_{i+2h+1-t}q^{i}+a_{2h-t}\sum_{i=0}^{2h-1}(-1)^{i+1}q^{i}+\sum_{i=t}^{2h-1}a_{i-t}q^{i}+\sum_{i=0}^{2h}(-1)^{i+1}q^{i}\\
		&=(a_{2h-t}+a_{2h-1-t}-q+1)q^{2h-1}+(a_{2h-t}+1)\sum_{i=0}^{2h-2}(-1)^{i+1}q^{i}+\sum_{i=t}^{2h-2}a_{i-t}q^{i}\\
		&\quad -\sum_{i=0}^{t-h-1}a_{i+2h+1-t}q^{i}.
	\end{align*}
	It is easy to verify that $0<[aq^{t}]_{n}<a$ requires $a_{2h-t}+a_{2h-1-t}=q-1$, and so
	$$[aq^{t}]_{n}=(a_{2h-t}+1)\sum_{i=0}^{2h-2}(-1)^{i+1}q^{i}+\sum_{i=t}^{2h-2}a_{i-t}q^{i}-\sum_{i=0}^{t-h-1}a_{i+2h+1-t}q^{i}.$$
	
		{\bf Subcase 3.1.} Suppose $t$ is even. Then
	\begin{align*}
		[aq^{t}]_{n}&=\sum_{i=t+1}^{2h-2}\big[(-1)^{i+1}(a_{2h-t}+1)+a_{i-t}\big]q^{i}+(a_{2h-t}+1)\sum_{i=0}^{t}(-1)^{i+1}q^{i}+a_{0}q^{t}\\
		&\quad -\sum_{i=0}^{t-h-1}a_{i+2h+1-t}q^{i}\\
		&=\sum_{i=(t+2)/2}^{h-1}\Delta_{i,t}q^{2i-1}-(a_{2h-t}+1)\frac{q^{t+1}+1}{q+1}+a_{0}q^{t}-\sum_{i=0}^{t-h-1}a_{i+2h+1-t}q^{i},
	\end{align*}
	where $\Delta_{i,t}=(-a_{2h-t}-1+a_{2i-t})q+a_{2h-t}+1+a_{2i-1-t}$. Similar to the proof in Case 2, we can prove that $0<[aq^{t}]_{n}<a$ only if $\Delta_{i,t}=0$ for $\frac{t+2}{2}\leq i\leq h-1$. Therefore,
	\begin{align*}
		[aq^{t}]_{n}&=-(a_{2h-t}+1)\frac{q^{t+1}+1}{q+1}+a_{0}q^{t}-\sum_{i=0}^{t-h-1}a_{i+2h+1-t}q^{i}\\
		&=\frac{(a_{0}-a_{2h-t}-1)q^{t+1}+a_{0}q^{t}-\sum\limits_{i=0}^{t-h-1}a_{i+2h+1-t}q^{i}(q+1)-a_{2h-t}-1}{q+1}.
	\end{align*}
	In order to make sure $[aq^{t}]_{n}>0$, we must have $a_{0}-a_{2h-t}-1\geq 0$.
	
	 If $t=h+1$, we have
	$$[aq^{h+1}]_{n}=\frac{(a_{0}-a_{h-1}-1)q^{h+2}+a_{0}q^{h+1}-a_{h}q-a_{h}-a_{h-1}-1}{q+1}.$$
	If $a_{0}-a_{h-1}-1\geq 1$, then $[aq^{h+1}]_{n}\geq \frac{q^{h+2}+q^{h+1}-q^{2}+2}{q+1}>(q-1)q^{h}>a$; so $a_{0}-a_{h-1}-1=0$ implying $a_{0}=a_{h-1}+1$ and
	$$[aq^{h+1}]_{n}=\frac{(a_{h-1}+1)q^{h+1}-a_{h}q-a_{h}-a_{h-1}-1}{q+1}.$$
	We know from above that $a_{h-1}=a_{2i-h-1}$ and $q-1-a_{h-1}=a_{h-2}=a_{2i-h-2}$ for $\frac{h+3}{2}\leq i\leq h-1$; thus
	\begin{align*}
	a&=a_{h}q^{h}+a_{h-1}(q^{h-1}+q^{h-3}+\cdots+q^{2})+(q-1-a_{h-1})(q^{h-2}+q^{h-4}+\cdots+q)+a_{h-1}+1\\
	&=\frac{a_{h}q^{h+1}+(a_{h}+a_{h-1}+1)q^{h}+a_{h-1}+1}{q+1}.
	\end{align*}
	It follows that
	\begin{align*}
		[aq^{h+1}]_{n}<a~&\Leftrightarrow~(a_{h-1}+1-a_{h})q^{h+1}-(a_{h}+a_{h-1}+1)q^{h}-a_{h}q-a_{h}-2a_{h-1}-2<0\\
		&\Leftrightarrow~(a_{h-1}+1-a_{h})q\leq a_{h}+a_{h-1}+1\\
		&\Leftrightarrow~a_{h-1}+1-a_{h}=1~{\rm and}~q\leq a_{h}+a_{h-1}+1,~{\rm or}~a_{h-1}+1-a_{h}\leq 0\\
		&\Leftrightarrow~a_{h-1}=a_{h}\geq \lceil\frac{q-1}{2}\rceil,~{\rm or}~a_{h-1}\leq a_{h}-1.
	\end{align*}
	
	 If $t\geq h+3$, we have
	\begin{align*}
		[aq^{t}]_{n}&\geq \frac{q^{t}-\sum\limits_{i=0}^{t-h-1}(q-1)q^{i}(q+1)-q}{q+1}=\frac{q^{t}-q^{t-h+1}-q^{t-h}+1}{q+1}=\frac{q^{t-h}(q^{h}-q-1)+1}{q+1}\\
		&\geq \frac{q^{3}(q^{h}-q-1)+1}{q+1}=\frac{q^{h+3}-q^{4}-q^{3}+1}{q+1}>(q-1)q^{h}>a.
	\end{align*}
	
	{\bf Subcase 3.2.} Suppose $t$ is odd. Then
	\begin{align*}
		[aq^{t}]_{n}&=\sum_{i=t}^{2h-2}\big[(-1)^{i+1}(a_{2h-t}+1)+a_{i-t}\big]q^{i}+(a_{2h-t}+1)\sum_{i=0}^{t-1}(-1)^{i+1}q^{i}-\sum_{i=0}^{t-h-1}a_{i+2h+1-t}q^{i}\\
		&=\sum_{i=(t+1)/2}^{h-1}\Delta_{i,t}q^{2i-1}-(a_{2h-t}+1)\frac{q^{t}+1}{q+1}-\sum_{i=0}^{t-h-1}a_{i+2h+1-t}q^{i},
	\end{align*}
	where $\Delta_{i,t}=(-a_{2h-t}-1+a_{2i-t})q+a_{2h-t}+1+a_{2i-1-t}$. Similarly as in the proof of Case 2, $0<[aq^{t}]_{n}<a$ only if $\Delta_{i,t}=0$ for $\frac{t+3}{2}\leq i\leq h-1$, and hence
	$$[aq^{t}]_{n}=\Delta_{\frac{t+1}{2},t}q^{t}-(a_{2h-t}+1)\frac{q^{t}+1}{q+1}-\sum_{i=0}^{t-h-1}a_{i+2h+1-t}q^{i}.$$
	It is easy to see that $[aq^{t}]_{n}>0$ if and only if $\Delta_{\frac{t+1}{2},t}>0$.
	
	 If $t= h+2$, then 
	$$[aq^{h+2}]_{n}=\Delta_{\frac{h+3}{2},h+2}q^{h+2}-(a_{h-2}+1)\frac{q^{h+2}+1}{q+1}-a_{h}q-a_{h-1},$$
	where $\Delta_{\frac{h+3}{2},h+2}=(-a_{h-2}-1+a_{1})q+a_{h-2}+1+a_{0}$. If $\Delta_{\frac{h+3}{2},h+2}\geq 2$, we have
	\begin{align*}
		[aq^{h+2}]_{n}&\geq 2q^{h+2}-\frac{q(q^{h+2}+1)}{q+1}-(q-2)q-q+1\\
		&=\frac{q^{h+3}+2q^{h+2}-q^{3}+q+1}{q+1}>(q-1)q^{h}>a.
		\end{align*}
	So $\Delta_{\frac{h+3}{2},h+2}=1$, and thus
	\begin{align*}
		[aq^{h+2}]_{n}&=q^{h+2}-(a_{h-2}+1)\frac{q^{h+2}+1}{q+1}-a_{h}q-a_{h-1}\\
		&=\frac{q^{h+3}-a_{h-2}q^{h+2}-a_{h}q^{2}-a_{h}q-a_{h-1}(q+1)-a_{h-2}-1}{q+1}.
	\end{align*}
	If $a_{h-2}\leq q-2$, then 
		\begin{align*}
		[aq^{h+2}]_{n}&\geq \frac{ 2q^{h+2}-(q-2)q^{2}-(q-2)q-(q-1)(q+1)-q+1}{q+1}\\
		&=\frac{2q^{h+2}-q^{3}+q+2}{q+1}>(q-1)q^{h}>a,
	\end{align*}
	which forces $a_{h-2}=q-1$, and therefore
	$$[aq^{h+2}]_{n}=\frac{q^{h+2}-a_{h}q^{2}-(a_{h}+a_{h-1}+1)q-a_{h-1}}{q+1}.$$
	However, as $a_{h}\leq q-2$ we have
	$$[aq^{h+2}]_{n}\geq \frac{q^{h+2}-(q-2)q^{2}-(2q-2)q-q+1}{q+1}=\frac{q^{h+2}-q^{3}+q+1}{q+1}>(q-1)q^{h}>a.$$
	
	 If $t\geq h+4$, then 
	\begin{align*}
		[aq^{t}]_{n}&\geq q^{t}-\frac{q(q^{t}+1)}{q+1}-\sum_{i=0}^{t-h-1}(q-1)q^{i}=\frac{q^{t}-q^{t-h+1}-q^{t-h}+1}{q+1}=\frac{q^{t-h}(q^{h}-q-1)+1}{q+1}\\
		&\geq \frac{q^{4}(q^{h}-q-1)+1}{q+1}=\frac{q^{h+4}-q^{5}-q^{4}+1}{q+1}>(q-1)q^{h}>a.
	\end{align*}
	
	In conclusion, for $h+1\leq t\leq 2h-2$, $[aq^{t}]_{n}<a$ if and only if $t=h+1$ and $a$ satisfies one of the following two conditions:
	
	\noindent {\bf c2)} $a=a_{h}q^{h}+(a_{h}+1)\frac{q^{h}+1}{q+1}$ with $\lceil\frac{q-1}{2}\rceil\leq a_{h}\leq q-2$,
	
   \noindent {\bf c3)} $a=a_{h}q^{h}+(a_{h-1}+1)\frac{q^{h}+1}{q+1}$ with $0\leq a_{h-1}<a_{h}\leq q-2$.
	
	{\bf Case 4.} When $t=2h-1$, we have
	\begin{align*}
		aq^{2h-1}&=\big(\sum_{i=0}^{h}a_{i}q^{i}\big)q^{2h-1}=\sum_{i=2h-1}^{3h-1}a_{i-2h+1}q^{i}=\sum_{i=2h+1}^{3h-1}a_{i-2h+1}q^{i}+a_{1}q^{2h}+a_{0}q^{2h-1}\\
		&\equiv -\sum_{i=0}^{h-2}a_{i+2}q^{i}+a_{1}\sum_{i=0}^{2h-1}(-1)^{i+1}q^{i}+a_{0}q^{2h-1}~({\rm mod}~n).
	\end{align*}
	Denote $M_{2h-1}=-\sum\limits_{i=0}^{h-2}a_{i+2}q^{i}+a_{1}\sum\limits_{i=0}^{2h-1}(-1)^{i+1}q^{i}+a_{0}q^{2h-1}$. It is easy to verify that $0<M_{2h-1}<2n$ and if $0<M_{2h-1}<n$, then $[aq^{2h-1}]_{n}=M_{2h-1}>(q-1)q^{h}>a$. Suppose $n<M_{2h-1}<2n$. Then
	\begin{align*}
		[aq^{2h-1}]_{n}&=M_{2h-1}-n=-\sum_{i=0}^{h-2}a_{i+2}q^{i}+a_{1}\sum_{i=0}^{2h-1}(-1)^{i+1}q^{i}+a_{0}q^{2h-1}+\sum_{i=0}^{2h}(-1)^{i+1}q^{i}\\
		&=(a_{1}+a_{0}-q+1)q^{2h-1}-(a_{1}+1)\frac{q^{2h-1}+1}{q+1}-\sum_{i=0}^{h-2}a_{i+2}q^{i}.
	\end{align*}
	If $a_{1}+a_{0}\leq q-1$, then $[aq^{2h-1}]_{n}<0$, which is impossible. If $a_{1}+a_{0}>q-1$, we have
	\begin{align*}
		[aq^{2h-1}]_{n}&\geq q^{2h-1}-\frac{q(q^{2h-1}+1)}{q+1}-(q-2)q^{h-2}-\sum_{i=0}^{h-3}(q-1)q^{i}\\
		&=\frac{q^{2h-1}-q^{h}+q^{h-2}+1}{q+1}>(q-1)q^{h}>a.
		\end{align*}
	
	{\bf Case 5.} When $t=2h$, we have
	\begin{align*}
		aq^{2h}&=\big(\sum_{i=0}^{h}a_{i}q^{i}\big)q^{2h}=\sum_{i=2h}^{3h}a_{i-2h}q^{i}=\sum_{i=2h+1}^{3h}a_{i-2h}q^{i}+a_{0}q^{2h}\\
		&\equiv -\sum_{i=0}^{h-1}a_{i+1}q^{i}+a_{0}\sum\limits_{i=0}^{2h-1}(-1)^{i+1}q^{i}~({\rm mod}~n).
	\end{align*}
	Denote $M_{2h}=-\sum\limits_{i=0}^{h-1}a_{i+1}q^{i}+a_{0}\sum\limits_{i=0}^{2h-1}(-1)^{i+1}q^{i}$. One has 
	$$M_{2h}<(q-1)\sum_{i=0}^{2h-1}(-1)^{i+1}q^{i}= \frac{(q-1)(q^{2h}-1)}{q+1}=\frac{q^{2h+1}-q^{2h}-q+1}{q+1}<n,$$ and
	$$M_{2h}\geq -(q-2)q^{h-1}-\sum_{i=0}^{h-2}(q-1)q^{i}+\frac{q^{2h}-1}{q+1}=\frac{q^{2h}-q^{h+1}+q^{h-1}+q}{q+1}>(q-1)q^{h}>a.$$
	
	{\bf Case 6.} When $2h+1\leq t\leq 3h$, we have $aq^{t}\equiv -aq^{t-2h-1}~({\rm mod}~n)$. Note that $$0<aq^{t-2h-1}\leq (q-1)q^{h}\cdot q^{t-2h-1}=(q-1)q^{t-h-1}\leq (q-1)q^{2h-1}<n.$$ So $[aq^{t}]_{n}=n-aq^{t-2h-1}\geq \frac{q^{2h+1}+1}{q+1}-(q-1)q^{2h-1}=\frac{q^{2h-1}+1}{q+1}>(q-1)q^{h}>a$.
	
	{\bf Case 7.} When $t=3h+1$, we know from the proof of Case 2 that $aq^{3h+1}\equiv -aq^{h}\equiv-M_{h}~({\rm mod}~n)$, where $M_{h}=a_{h}\sum\limits_{i=0}^{2h-1}(-1)^{i+1}q^{i}+\sum\limits_{i=h}^{2h-1}a_{i-h}q^{i}$ and $0<M_{h}<2n$. If $n<M_{h}<2n$, we have
	\begin{align*}
		[aq^{3h+1}]_{n}&=2n-M_{h}=\frac{2q^{2h+1}+2}{q+1}-a_{h}\sum_{i=0}^{2h-1}(-1)^{i+1}q^{i}-\sum_{i=h}^{2h-1}a_{i-h}q^{i}\\
		&\geq\frac{2q^{2h+1}+2}{q+1}-(q-2)\sum_{i=0}^{2h-1}(-1)^{i+1}q^{i}-\sum_{i=h}^{2h-1}(q-1)q^{i}\\
		&=\frac{q^{2h}+q^{h+1}+q^{h}+q}{q+1}>(q-1)q^{h}>a.
	\end{align*}
	Suppose $0<M_{h}<n$. Then 
	\begin{align*}
		[aq^{3h+1}]_{n}&=n-M_{h}=-(M_{h}-n)\\
		&=(q-1-a_{h}-a_{h-1})q^{2h-1}+(a_{h}+1)\sum_{i=0}^{2h-2}(-1)^{i}q^{i}-\sum_{i=h}^{2h-2}a_{i-h}q^{i}\\
		&=(q-1-a_{h}-a_{h-1})q^{2h-1}+\sum_{i=h}^{2h-2}\big[(-1)^{i}(a_{h}+1)-a_{i-h}\big]q^{i}+(a_{h}+1)\sum_{i=0}^{h-1}(-1)^{i}q^{i}\\
		&=(q-1-a_{h}-a_{h-1})q^{2h-1}+\sum_{i=(h+1)/2}^{h-1}\Delta_{i,3h+1}q^{2i-1}+(a_{h}+1)\frac{q^{h}+1}{q+1},
	\end{align*}
	where $\Delta_{i,3h+1}=(a_{h}+1-a_{2i-h})q-a_{h}-1-a_{2i-1-h}$. It is easy to verify that $0<[aq^{3h+1}]_{n}<a$ only if $a_{h}+a_{h-1}=q-1$ and $\Delta_{i,3h+1}=0$, or equivalently, $a_{2i-h}=a_{h}$ and $a_{2i-1-h}=q-1-a_{h}$, for $\frac{h+3}{2}\leq i\leq h-1$. Hence
	$$[aq^{3h+1}]_{n}=\Delta_{\frac{h+1}{2},3h+1}q^{h}+(a_{h}+1)\frac{q^{h}+1}{q+1},$$
	where $\Delta_{\frac{h+1}{2},3h+1}=(a_{h}+1-a_{1})q-a_{h}-1-a_{0}$, and
	\begin{align*}
		a&=a_{h}(q^{h}+q^{h-2}+\cdots+q^{3})+(q-1-a_{h})(q^{h-1}+q^{h-3}+\cdots+q^{2})+a_{1}q+a_{0}\\
		&=(a_{h}+1)\frac{q^{h+1}-q^{2}}{q+1}+a_{1}q+a_{0}.
	\end{align*}
	 If $\Delta_{\frac{h+1}{2},3h+1}< 0$, then $[aq^{3h+1}]_{n}\leq-q^{h}+\frac{(q-1)(q^{h}+1)}{q+1}=-\frac{2q^{h}-q-1}{q+1} <0$, which is impossible. So $\Delta_{\frac{h+1}{2},3h+1}\geq0$, that is, $(a_{h}+1-a_{1})q\geq a_{h}+1+a_{0}$, which is equivalent to one of the following:
	 
	 \noindent (i) $a_{h}+1-a_{1}=1$ and $q\geq a_{h}+1+a_{0}$,
	 
	 \noindent (ii) $a_{h}+1-a_{1}\geq 2$.
	
	Assume (i), then $a_{1}=a_{h}$ and $a_{h}+a_{0}\leq q-1$. Thus
    $$[aq^{3h+1}]_{n}=(q-a_{h}-1-a_{0})q^{h}+(a_{h}+1)\frac{q^{h}+1}{q+1}=\frac{(q-a_{h}-a_{0})q^{h+1}-a_{0}q^{h}+a_{h}+1}{q+1},$$
	and
	$$a=(a_{h}+1)\frac{q^{h+1}-q^{2}}{q+1}+a_{h}q+a_{0}=\frac{(a_{h}+1)q^{h+1}-q^{2}+(a_{h}+a_{0})q+a_{0}}{q+1}.$$
	It follows that
	\begin{align*}
		[aq^{3h+1}]_{n}<a~&\Leftrightarrow~(q-2a_{h}-a_{0}-1)q^{h+1}-a_{0}q^{h}+q^{2}-(a_{h}+a_{0})q+a_{h}-a_{0}+1<0\\
		&\Leftrightarrow~(q-2a_{h}-a_{0}-1)q-a_{0}<0\\
		&\Leftrightarrow~q-2a_{h}-a_{0}-1\leq 0.
	\end{align*}
	
	Assume (ii), then $a_{h}-a_{1}\geq 1$. If $a_{h}-a_{1}\geq 2$, then $\Delta_{\frac{h+1}{2},3h+1}\geq 3q-(q-2)-1-(q-1)=q+2$; thus $[aq^{3h+1}]_{n}\geq (q+2)q^{h}+\frac{3(q^{h}+1)}{q+1}>(q-1)q^{h}>a$. So $a_{h}-a_{1}=1$, therefore
	\begin{align*}
		[aq^{3h+1}]_{n}&=(2q-a_{h}-1-a_{0})q^{h}+(a_{h}+1)\frac{q^{h}+1}{q+1}\\
		&=\frac{(2q+1-a_{h}-a_{0})q^{h+1}-a_{0}q^{h}+a_{h}+1}{q+1}.
	\end{align*}
	and
	$$a=(a_{h}+1)\frac{q^{h+1}-q^{2}}{q+1}+(a_{h}-1)q+a_{0}=\frac{(a_{h}+1)q^{h+1}-2q^{2}+(a_{h}+a_{0}-1)q+a_{0}}{q+1}.$$
	We then have
	\begin{align*}
		[aq^{3h+1}]_{n}<a~&\Leftrightarrow~(2q-2a_{h}-a_{0})q^{h+1}-a_{0}q^{h}+2q^{2}-(a_{h}+a_{0}-1)q+a_{h}-a_{0}+1<0\\
		&\Leftrightarrow~(2q-2a_{h}-a_{0})q-a_{0}<0\\
		&\Leftrightarrow~2q-2a_{h}-a_{0}\leq 0.
	\end{align*}
	
	In conclusion, $[aq^{3h+1}]_{n}<a$ if and only if $a$ satisfies one of the following two conditions:
	
	\noindent {\bf c4)} $a=(a_{h}+1)\frac{q^{h+1}+q}{q+1}-q+a_{0}$ with $1\leq a_{h}\leq q-2$ and ${\rm max}\{1,q-1-2a_{h}\}\leq a_{0}\leq q-1-a_{h}$,
	
	\noindent {\bf c5)} $a=(a_{h}+1)\frac{q^{h+1}+q}{q+1}-2q+a_{0}$ with $\lceil\frac{q+1}{2}\rceil \leq a_{h}\leq q-2$ and $2q-2a_{h}\leq a_{0}\leq q-1$.
	
	{\bf Case 8.} When $3h+2\leq t\leq 4h-1$, let $\ell=t-2h-1$, then $h+1\leq \ell \leq 2h-2$, and it follows from the proof of Case 3 that $$aq^{t}\equiv-aq^{\ell}\equiv -M_{\ell}~({\rm mod}~n),$$ where $M_{\ell}=-\sum\limits_{i=0}^{\ell-h-1}a_{i+2h+1-\ell}q^{i}+a_{2h-\ell}\sum\limits_{i=0}^{2h-1}(-1)^{i+1}q^{i}+\sum\limits_{i=\ell}^{2h-1}a_{i-\ell}q^{i}$ and $0<M_{\ell}<2n$. If $n<M_{\ell}<2n$, one can check that $[aq^{t}]_{n}=2n-M_{\ell}>(q-1)q^{h}>a$. Suppose $0<M_{\ell}<n$. Then
	\begin{align*}
		[aq^{t}]_{n}&=n-M_{\ell}=-(M_{\ell}-n)\\
		&=(q-1-a_{2h-\ell}-a_{2h-1-\ell})q^{2h-1}+(a_{2h-\ell}+1)\sum_{i=0}^{2h-2}(-1)^{i}q^{i}-\sum_{i=\ell}^{2h-2}a_{i-\ell}q^{i}\\
		&\quad +\sum_{i=0}^{\ell-h-1}a_{i+2h+1-\ell}q^{i}.
	\end{align*}
	Direct calculation shows that $0<[aq^{t}]_{n}<a$ only if $a_{2h-\ell}+a_{2h-1-\ell}=q-1$, and thus
	$$[aq^{t}]_{n}=(a_{2h-\ell}+1)\sum_{i=0}^{2h-2}(-1)^{i}q^{i}-\sum_{i=\ell}^{2h-2}a_{i-\ell}q^{i}+\sum_{i=0}^{\ell-h-1}a_{i+2h+1-\ell}q^{i}.$$
	
	{\bf Subcase 8.1.} Suppose $t$ is odd, then $\ell$ is even, and
	$$[aq^{t}]_{n}=\sum_{i=(\ell+2)/2}^{h-1}\Delta_{i,t}q^{2i-1}+(a_{2h-\ell}+1)\frac{q^{\ell+1}+1}{q+1}-a_{0}q^{\ell}+\sum_{i=0}^{\ell-h-1}a_{i+2h+1-\ell}q^{i},$$
	where $\Delta_{i,t}=(a_{2h-\ell}+1-a_{2i-\ell})q-a_{2h-\ell}-1-a_{2i-1-\ell}$.
	Similarly as in the proof of Case 2, we can prove that $0<[aq^{t}]_{n}<a$ requires $\Delta_{i,t}=0$ for $\frac{\ell+2}{2}\leq i\leq h-1$. Hence
	\begin{align*}
		[aq^{t}]_{n}&=(a_{2h-\ell}+1)\frac{q^{\ell+1}+1}{q+1}-a_{0}q^{\ell}+\sum_{i=0}^{\ell-h-1}a_{i+2h+1-\ell}q^{i}\\
		&=\frac{(a_{2h-\ell}+1-a_{0})q^{\ell+1}-a_{0}q^{\ell}+\sum\limits_{i=0}^{\ell-h-1}a_{i+2h+1-\ell}q^{i}(q+1)+a_{2h-\ell}+1}{q+1}.
	\end{align*}
	One can check that $[aq^{t}]_{n}>0$ if and only if $a_{2h-\ell}+1-a_{0}\geq 1$. 
	
	 If $t=3h+2$, then $\ell=h+1$ and 
	$$[aq^{3h+2}]_{n}=\frac{(a_{h-1}+1-a_{0})q^{h+2}-a_{0}q^{h+1}+a_{h}(q+1)+a_{h-1}+1}{q+1}.$$
	If $a_{h-1}+1-a_{0}\geq 2$, we have
	$$[aq^{3h+2}]_{n}\geq \frac{2q^{h+2}-(q-1)q^{h+1}+q+1+2+1}{q+1}=\frac{q^{h+2}+q^{h+1}+q+4}{q+1}>(q-1)q^{h}>a.$$
	So $a_{h-1}+1-a_{0}=1$, and therefore
	$$[aq^{3h+2}]_{n}=\frac{(q-a_{h-1})q^{h+1}+a_{h}q+a_{h}+a_{h-1}+1}{q+1}.$$
	We know from above that $a_{h-1}=a_{h-3}=\cdots=a_{2}=a_{0}$ and $q-1-a_{h-1}=a_{h-2}=a_{h-4}=\cdots=a_{1}$; thus
	\begin{align*}
		a&=a_{h}q^{h}+a_{h-1}(q^{h-1}+q^{h-3}+\cdots+q^{2}+1)+(q-1-a_{h-1})(q^{h-2}+q^{h-4}+\cdots+q)\\
		&=\frac{a_{h}q^{h+1}+(a_{h}+a_{h-1}+1)q^{h}-q+a_{h-1}}{q+1}.
	\end{align*}
	Then we have
	\begin{align*}
		[aq^{3h+2}]_{n}<a~&\Leftrightarrow~(q-a_{h}-a_{h-1})q^{h+1}-(a_{h}+a_{h-1}+1)q^{h}+(a_{h}+1)q+a_{h}+1<0\\
		&\Leftrightarrow~(q-a_{h}-a_{h-1})q-a_{h}-a_{h-1}-1<0\\
		&\Leftrightarrow~a_{h}+a_{h-1}\geq q.
	\end{align*}
	
    If $t\geq 3h+4$, then $\ell=t-2h-1\geq h+3$ and
	$$[aq^{t}]_{n}\geq \frac{q^{\ell+1}-(q-1)q^{\ell}+1}{q+1}=\frac{q^{\ell}+1}{q+1}\geq \frac{q^{h+3}+1}{q+1}>(q-1)q^{h}>a.$$
	
	{\bf Subcase 8.2.} Suppose $t$ is even, then $\ell$ is odd, and
	$$[aq^{t}]_{n}=\sum_{i=(\ell+1)/2}^{h-1}\Delta_{i,t}q^{2i-1}+(a_{2h-\ell}+1)\frac{q^{\ell}+1}{q+1}+\sum_{i=0}^{\ell-h-1}a_{i+2h+1-\ell}q^{i},$$
	where $\Delta_{i,t}=(a_{2h-\ell}+1-a_{2i-\ell})q-a_{2h-\ell}-1-a_{2i-1-\ell}$. Similar to the proof of Case 2, $0<[aq^{t}]_{n}<a$ only if $\Delta_{i,t}=0$ for $\frac{\ell+1}{2}\leq i\leq h-1$. Thus
	$$[aq^{t}]_{n}=(a_{2h-\ell}+1)\frac{q^{\ell}+1}{q+1}+\sum_{i=0}^{\ell-h-1}a_{i+2h+1-\ell}q^{i}.$$
	Since $t\geq 3h+3$, $\ell=t-2h-1\geq h+2$, and then $[aq^{t}]_{n}\geq \frac{q^{\ell}+1}{q+1}\geq \frac{q^{h+2}+1}{q+1}>(q-1)q^{h}>a$.

	Hence for $3h+2\leq t\leq 4h-1$, $[aq^{t}]_{n}<a$ if and only if $t=3h+2$ and $a$ satisfies the following condition:
	
	\noindent {\bf c6)} $a=a_{h}q^{h}+(a_{h-1}+1)\frac{q^{h}+1}{q+1}-1$ with $1\leq a_{h}\leq q-2$ and $q-a_{h}\leq a_{h-1}\leq q-1$.
	
	{\bf Case 9.} When $t=4h$, we see from Case 4 that
	$$aq^{4h}\equiv -aq^{2h-1}\equiv -M_{2h-1}~({\rm mod}~n).$$
	where $M_{2h-1}=-\sum\limits_{i=0}^{h-2}a_{i+2}q^{i}+a_{1}\sum\limits_{i=0}^{2h-1}(-1)^{i+1}q^{i}+a_{0}q^{2h-1}$ and $0<M_{2h-1}<2n$. If $n<M_{2h-1}<2n$, one can check that $[aq^{4h}]_{n}=2n-M_{2h-1}>(q-1)q^{h}>a$. Suppose $0<M_{2h-1}<n$. Then
	\begin{align*}
		[aq^{4h}]_{n}&=n-M_{2h-1}=-(M_{2h-1}-n)\\
		&=(q-1-a_{1}-a_{0})q^{2h-1}+(a_{1}+1)\frac{q^{2h-1}+1}{q+1}+\sum_{i=0}^{h-2}a_{i+2}q^{i}.
	\end{align*}
	If $a_{0}+a_{1}> q-1$, then
	$$[aq^{4h}]_{n}\leq -q^{2h-1}+\frac{q(q^{2h-1}+1)}{q+1}+\sum_{i=0}^{h-2}(q-1)q^{i}=-\frac{q^{2h-1}-q^{h}-q^{h-1}+1}{q+1}<0,$$
	which is impossible. If $a_{0}+a_{1}\leq q-1$, we have $[aq^{4h}]_{n}>\frac{q^{2h-1}+1}{q+1}>(q-1)q^{h}>a$. 
	
	{\bf Case 10.} When $t=4h+1$, it follows from Case 5 that $aq^{t}\equiv -aq^{2h}\equiv -M_{2h}~({\rm mod}~n)$,
	where $M_{2h}=-\sum\limits_{i=0}^{h-1}a_{i+1}q^{i}+a_{0}\sum\limits_{i=0}^{2h-1}(-1)^{i+1}q^{i}$ and $0<M_{2h}<\frac{q^{2h+1}-q^{2h}-q+1}{q+1}<n$. Then $$[aq^{t}]_{n}=n-M_{2h}>\frac{q^{2h}+q}{q+1}>(q-1)q^{h}>a.$$
	
	Collecting all the conclusions above yields the desired result.
	\qed


\begin{thebibliography}{99}
		\bibitem{bb} A. Betten, M. Braun, H. Fripertinger, A. Kerber, A. Kohnert, and A. Wassermann, Error-Correcting Linear Codes: Classification by Isometry and Applications. Berlin, Germany: Springer-Verlag, 2006.
		
		 \bibitem{bccgh} J. Bringer, C. Carlet, H. Chabanne, S. Guilley, and H. Maghrebi, Orthogonal direct sum masking, a smartcard friendly computation paradigm in a code, with builtin protection against side-channel and fault attacks. {\em Proc. WISTP}, 2014, 40-56.
		 
		\bibitem{cg} C. Carlet, S. Guilley, Complementary dual codes for counter-measures to side-channel attacks. {\em Adv. Math.
		Commun.}, 2016, {\bf 10}(1): 131-150.
		
		\bibitem{gdl} B. Gong, C. Ding, C. Li, The dual codes of several classes of BCH codes. {\em IEEE Trans. Inf. Theory}, 2021, {\bf 68}(2): 953-96.
		
		\bibitem{hywsm} X. Huang, Q. Yue, Y. Wu, X. Shi, J. Michel, Binary primitive LCD BCH codes. {\em Des. Codes Cryptogr.}, 2020, {\bf 88}(12): 2453-2473.
		
		\bibitem{hyws} X. Huang, Q. Yue, Y. Wu, X. Shi, Ternary primitive LCD BCH codes. {\em Adv. Math.
			Commun.}, 2023, {\bf 17}(3): 644-659.
		
		\bibitem{ldl1} C. Li, C. Ding, S. Li, LCD cyclic codes over finite fields. {\em IEEE Trans. Inf. Theory}, 2017, {\bf 63}(7): 4344-4356.
		
		\bibitem{lldl} S. Li, C. Li, C. Ding, H. Liu, Two families of LCD BCH codes. {\em IEEE Trans. Inf. Theory}, 2017, {\bf 63}(9): 5699-5717.
			
		\bibitem{ldl2} H. Liu, C. Ding, C. Li, Dimensions of three types of BCH codes over GF($q$). {\em Discrete Math.}, 2017, {\bf 340}: 1910-1927.
			
		\bibitem{llgs} Y. Liu, R. Li, L. Guo, H. Song, Dimensions of nonbinary antiprimitive BCH codes and some conjectures. {\em Discrete Math.}, 2023, https://doi.org/10.1016/j.disc.2023.113496
			
		\bibitem{llfl} Y. Liu, Y. Li, Q. Fu, L. Lu, Y. Rao, Some binary BCH codes with length $n=2^m+1$. {\em Finite Fields Appl.}, 2019, {\bf 55}: 109-133.
		
		\bibitem{m1} J. L. Massey, Reversible codes. {\em Inf. Control}, 1964, {\bf 7}(3): 369-380.
			
		\bibitem{wwlw} X. Wang, J. Wang, C. Li, Y. W, Two classes of narrow-sense BCH codes and their duals. 2022, arXiv:2210.08463v1 [cs.IT].
			
		\bibitem{ylly} H. Yan, H. Liu, C. Li, S. Yang, Parameters of LCD BCH codes with two lengths. {\em Adv. Math. Commun.}, 2018, {\bf 12}(3): 579-594.
		
		\bibitem{zc} H. Zhang, X. Cao, Dimensions of some LCD BCH codes. 2023, arXiv:2305.06508 [cs.IT].
		
		\bibitem{zlks} J. Zhang, P. Li, X. Kai, Z. Sun, Some results on BCH codes of length $\frac{q^m+1}{2}$. {\em Adv. Math. Commun.}, 2023, doi: 10.3934/amc.2023010.
		
		\bibitem{z} Y. Zhang, A class of LCD BCH codes of length $n=\frac{q^m+1}{\lambda}$. 2023, arXiv:2305.11677 [cs.IT].
		
		\bibitem{zlx} Y. Zhang, L. Liu, X. Xie, Three classes of BCH codes and their duals. 2022, arXiv:2211.15350v2 [cs.IT].
		
		\bibitem{zswh} H. Zhu, M. Shi, X. Wang, T. Helleseth, The $q$-ary antiprimitive BCH codes. {\em IEEE Trans. Inf. Theory}, 2021, {\bf 68}(3): 1683-1695.
	\end{thebibliography}
\end{document}